\documentclass[11pt]{amsart}

\usepackage{url,bm,amssymb,amscd,amsmath}
\usepackage{graphics,epsfig,color}
\usepackage[letterpaper,asymmetric,left=1.0in,right=1.0in,top=1.0in,bottom=1.
25in,bindingoffset=0.0in]{geometry}
\usepackage{mathtools}
\usepackage{amssymb}
\usepackage{amsfonts}
\usepackage{amsmath}
\usepackage{amsthm}
\usepackage{listings}
\usepackage[utf8]{inputenc}
\usepackage[english]{babel}
\usepackage{hyperref}
\newtheorem*{theorem*}{Theorem}
\newtheorem{theorem}{Theorem}[section]
\newtheorem{lemma}[theorem]{Lemma}
\newtheorem{proposition}[theorem]{Proposition}

\newtheorem{corollary}[theorem]{Corollary}

\newtheorem{remark}[theorem]{Remark}

\newtheorem*{THMA}{Theorem A}
\newtheorem*{THMB}{Theorem B}
\newtheorem*{THMC}{Theorem C}
\newtheorem*{THMD}{Theorem D}

\newtheorem*{expansion}{Criterion for Expansion}
\newtheorem*{leeyang}{Lee--Yang Theorem}
\newtheorem*{Youngfmla}{Ledrappier--Young Formula}
\newtheorem*{Shubthm}{Shub's Theorem}
\newtheorem*{SSthm}{Shub--Sullivan Theorem}
\newtheorem*{uniquenessequilib}{Uniqueness of Equilibrium States}
\newtheorem*{Ruelleineq}{Ruelle's Inequality}
\newtheorem*{specialergodictheorem}{Special Ergodic Theorem}
\usepackage{tabularx,ragged2e,booktabs,caption}
\usepackage{enumerate}
\usepackage{enumitem}
\usepackage[numbers]{natbib}
\usepackage{caption}
\usepackage{subcaption}
\usepackage{mathrsfs}

\newcommand{\MEAS}{\rho}

\newcommand{\bzt}{B_{z,t}}
\newcommand{\phit}{(\phi,t)}
\newcommand{\T}{\mathbb{T}}
\newcommand{\D}{\mathbb{D}}
\newcommand{\CMD}{\hat{\mathbb{C}}\setminus\overline{\mathbb{D}}}
\newcommand{\C}{\hat{\mathbb{C}}}

\newcommand{\pzt}{\mathcal{P}_{z,t}}
\newcommand{\WD}{w_{\D}}
\newcommand{\WCMD}{w_{\CMD}}

\newcommand{\bphittilde}{\widetilde{B}_{\phi,t}}
\newcommand{\btilde}{\widetilde{B}}

\begin{document}

\title{Limiting Measure of Lee--Yang Zeros for the Cayley Tree}

\begin{author}[I. Chio]{Ivan Chio}
\email{ichio@iupui.edu}
\address{ %
IUPUI Department of Mathematical Sciences\\
LD Building, Room 255\\
402 North Blackford Street\\
Indianapolis, Indiana 46202-3267\\
 United States }
\end{author}

\begin{author}[C. He]{Caleb He}
\email{calebhe@college.harvard.edu}
\address{Harvard University, Cambridge, MA. USA.}
\end{author}

\begin{author}[A. L. Ji]{Anthony L. Ji}
\email{anthony.ji@yale.edu}
\address{Yale University, New Haven, CT. USA.}
\end{author}

\begin{author}[R. K. W. Roeder]{Roland K. W. Roeder}
\email{roederr@iupui.edu}
\address{ %
IUPUI Department of Mathematical Sciences\\
LD Building, Room 224Q\\
402 North Blackford Street\\
Indianapolis, Indiana 46202-3267\\
 United States }
\end{author}

\date{\today}

\begin{abstract}
This paper is devoted to an in-depth study of the limiting measure of
Lee--Yang zeroes for the Ising Model on the Cayley Tree.  We build on previous
works of M\"uller-Hartmann--Zittartz (1974 and 1977), Barata--Marchetti (1997), and
Barata--Goldbaum (2001), to determine the support of the limiting measure, prove that the limiting measure is
not absolutely continuous with respect to Lebesgue measure, and determine the
pointwise dimension of the measure at Lebesgue a.e.\ point on
the unit circle and every temperature.  The latter is related to the critical exponents for the 
phase transitions in the model as one crosses the unit circle at Lebesgue
a.e.\ point, providing a global version of the ``phase transition of continuous
order'' discovered by M\"uller-Hartmann--Zittartz.  The key techniques are from
dynamical systems because there is an explicit formula for the Lee--Yang zeros
of the finite Cayley Tree of level $n$ in terms of the $n$-th iterate of an
expanding Blaschke Product.  A subtlety arises because the conjugacies between Blaschke Products at different parameter values
are not absolutely continuous.
\end{abstract}

\maketitle


\section{Introduction}

We study the the limiting measure of Lee--Yang zeros for the infinite
Cayley tree, a finite approximation of which is shown in Figure \ref{fig:cayleytree} below.  

\begin{figure}[h]
\begin{center}
\includegraphics[scale=0.4]{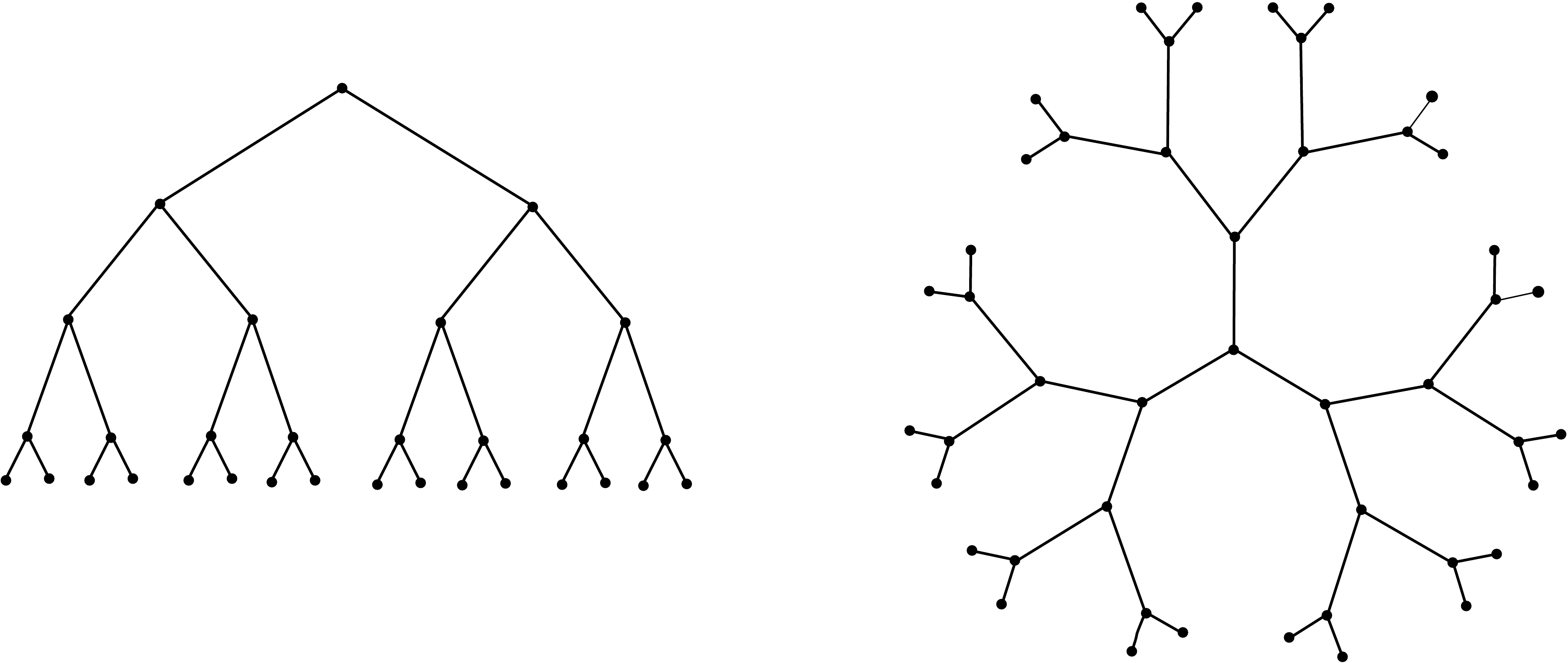}
\caption{Four levels of the Cayley tree with branching number $k = 2$.  The rooted version is shown on the left and the unrooted (full) version on the right.\label{fig:cayleytree}}
\end{center}
\end{figure}

\noindent
Consideration of the Lee--Yang zeros for the Ising Model on the Cayley Tree
dates back to works of M\"uller-Hartmann and Zittartz
\cite{MULLERZITTARZ,MULLERHARTMAN}, Barata--Marchetti \cite{BARATAMARCHETTI},
Barata--Goldbaum \cite{BARATAGOLDBAUM}, and others.  The hierarchical
structure of the Cayley Tree results in the following renormalization procedure
for studying the Lee--Yang zeros, which played a key role in each of the
aforementioned papers:

\begin{proposition}\label{PROP:RENORM}  For any $k \geq 2$, any $t \in [0,1)$ and any $z \in \mathbb{T}:=\{z \in  \mathbb{C} \, : \, |z| = 1\}$ consider the following Blaschke Product:
\begin{align}
B_{z,t,k}(w) := z\left(\frac{w+t}{1+wt}\right)^k.
\end{align}
The Lee--Yang zeros for the $n$-th rooted Cayley Tree with branching number $k\geq2$ are solutions $z$ to
  \begin{equation}
        B_{z,t,k}^{n}(z) = -1,
  \end{equation}
and the Lee--Yang zeros for the $n$-th full Cayley Tree with branching number $k\geq2$ are solutions $z$ to
  \begin{equation}
        B_{z,t,k+1} \circ B_{z,t,k}^{n-1}(z) = -1.
  \end{equation}
\label{lyz}
\end{proposition}
\noindent
Here, $z:={\rm exp}(-2h/T)$ and $t:={\rm exp}(-2J/T)$, where $h$ is the externally applied magnetic field, 
$T > 0$ is the temperature, and $J > 0$ is the coupling constant between neighboring atoms.
The superscript ``$n$'' denotes iteration of the function $n$ times.  When the exponent $k$ is clear from the context we will drop it
from the notation, writing $B_{z,t,k} \equiv B_{z,t}$.

Remark that many classical treatments of the Ising Model on the Cayley Tree
consider only the thermodynamical properties associated to vertices ``deep'' in
the lattice {\rm (}e.g. the root vertex{\rm )}; see {\rm \cite[Ch. 4]{BAXTER}}
and the references therein.  The term ``Bethe Lattice'' is customarily used to
describe such considerations.  Instead, we treat all vertices equally, studying
the ``bulk'' behavior of the lattice, and thus we follow the standard
convention of referring to our work as being on the Cayley Tree.

Before stating our results, we will give a brief background on Lee--Yang zeros,
including a description of what many people believe should hold for the
classical $\mathbb{Z}^d$ lattice (where $d \geq 2$), as well as a description of the previous
results of  M\"uller-Hartmann and Zittartz, Barata and Marchetti, and Barata and
Goldbaum.  The reader who already knows this background can skip ahead to Section~\ref{SEC:MAIN_RESULTS}.

Proposition \ref{PROP:RENORM} allows us to use powerful techniques from
dynamical systems to prove results about the Lee--Yang zeros for the Cayley
Tree, whose analogs are completely unknown for classical lattices like
$\mathbb{Z}^d$.  Therefore, our work lies at the boundary between dynamical
systems and statistical physics.  For this reason, we have attempted to provide
considerable background in both areas.

\subsection{Lee--Yang Zeros}
The Ising Model describes magnetic materials.  The matter at a certain scale
is described using a graph $\Gamma = (V,E)$ with vertex set $V$ and edge set $E$.
Here, $V$ represents atoms and $E$ represents the magnetic bonds between them.
Assign a spin to each vertex using a spin configuration
$\sigma:V\rightarrow\{\pm1\}$. The total energy of the configuration $\sigma$
is given as
\begin{equation}\label{EQN:HAMILTONIAN}
H(\sigma)=-J\cdot\sum_{\{v,w\}\in E}\sigma(v)\sigma(w)-h\cdot \sum_{v\in V}\sigma(v)\mbox{,}
\end{equation}
where $J>0$ is the coupling constant that describes the interaction between neighboring spins, and $h$ is the externally applied magnetic field.

The Boltzmann-Gibbs Principle gives that the probability $P(\sigma)$ of a
configuration $\sigma$ is proportional\footnote{We set the Boltzmann
constant $k_B=1$.} to $W(\sigma):=\exp(-H(\sigma)/T)$ for temperature $T>0$.
Explicitly, $P(\sigma)=W(\sigma)/Z\mbox{,}$ where $Z$ is the normalizing factor
defined as
\begin{align*}
Z \equiv Z(J,h,T):=\sum_\sigma{W(\sigma)},
\end{align*} which is summed over all
possible spin configurations~$\sigma$.  (Remark that we will always impose free boundary conditions on $\Gamma$.) This normalizing factor $Z$
is known as the \textit{partition function}.  It
is a fundamental quantity to study in statistical mechanics and most aggregate
thermodynamic quantities of a physical system can be derived from it. 

It is useful to make a change of variables to $z = \exp(-2h/T)$, which represents the magnetic field variable, and $t= \exp(-2J/T)$, which represents the temperature variable.
In these new variables, $Z(z,t)$ becomes a polynomial, if we multiply by $\sqrt{z}^{|V|} \sqrt{t}^{|E|}$ to clear the denominators.
For fixed $t\in [0,1]$, the behavior of $Z(z,t)$ can be fully understood by studying its complex zeros in the variable $z$. In 1952, T. D. Lee and C. N. Yang \cite{LEEYANG1} characterized these zeros, now known as {\em Lee--Yang zeros}, in their famous theorem.
\begin{leeyang}
For $t\in[0,1]$, the complex zeros in $z$ of the partition function $Z(z,t)$ for the Ising model on any graph lie on the unit circle $\mathbb{T}=\{|z|=1\}$.
\label{thm:ly}
\end{leeyang}

\noindent
Because of the Lee--Yang Theorem, 
throughout the paper we will refer to $z$ and $\phi:=\text{Arg}(z)$ interchangeably.

\subsection{Limiting Measure $\mu_t$ of Lee--Yang Zeros}
One typically describes a magnetic material at different scales using a sequence
of connected graphs $\Gamma_n = (V_n,E_n)$, each thought of as a finer approximation of
the material than the previous.  Let us call such a sequence of graphs a
``lattice''.  The standard example is the $\mathbb{Z}^d$ lattice where, for
each $n \geq 0$, one defines $\Gamma_n$ to be the graph whose vertices consist
of the integer points in $[-n,n]^d$ and whose edges connect vertices
at distance one in~$\mathbb{R}^d$.  

The physical properties of the magnetic material are described by limits of
suitably normalized thermodynamical quantities associated to each of the finite
graphs $\Gamma_n$.  Many of these can be
described in terms of the limiting measure of Lee--Yang zeros associated to
the lattice $\{\Gamma_n\}$, which we will now describe.  For each $n \geq 0$ let $Z_n(z,t)$ denote the
partition function associated to $\Gamma_n$ and let
$z_1(t),\ldots,z_{|V_n|}(t)$ denote
the Lee--Yang zeros at temperature $t \in [0,1]$.  For classical lattices
($\mathbb{Z}^d$, etc), it is a consequence of the van-Hove Theorem
\cite{VANHOVE} and the Lee--Yang Theorem that for each $t \in [0,1]$ the
sequence of measures
\begin{align*}
\mu_{t,n} := \frac{1}{|V_n|} \sum_{i=1}^{|V_n|} \delta_{z_i(t)}
\end{align*}
weakly converges to a limiting measure $\mu_t$ that is supported on the unit circle $\mathbb{T}$. 
One has the following expressions for the limiting free energy and magnetization:
\begin{equation}\label{electrostat rep}
    F(z,t)= - 2 T \int_\mathbb{T} \log |z-\zeta|\,  d\mu_t (\zeta)  +T\left(\, \log |z| + \left(\lim_{n\rightarrow \infty}
    \frac{|E_n|}{|V_n|}\right) {\log|t|}\, \right)
             \quad \mbox{for a.e. $z\in \mathbb{C}$},
\end{equation}
\begin{equation}\label{Cauchy rep}
     M(z,t) = -2\frac{\partial F}{\partial h} =   -4 z \int_\T  \frac {d\mu_t (\zeta) }{z-\zeta}+2 \quad \mbox{for $z\in \mathbb{C} \setminus {\rm Supp}(\mu_t)$}.
\end{equation}
See, for example, \cite[Prop. 2.2]{BLR1}.  However, note that a minor adaptation is needed
because in that paper the free energy is normalized by number of edges instead of number of vertices.

\subsection{Conjectural Description of $\mu_t$ for the $\mathbb{Z}^d$ lattice (where $d \geq 2$).}
A famous unsolved problem from statistical physics is to understand the
limiting measures of Lee--Yang zeros $\mu_t$ for the $\mathbb{Z}^d$ lattice and how they depend on $t$.
It is believed that for every $t \in [0,1)$ the measure $\mu_t$ is absolutely continuous with respect to Lebesgue measure $d\phi$ on the circle, and thus has density
$\rho_t(\phi) := \frac{d \mu_t}{d\phi}$.
Let $t_c > 0$ denote the critical temperature\footnote{More precisely, define $t_c$ 
to be the infimum of temperatures for which the spontaneous magnetization is zero.} of the $\mathbb{Z}^d$ Ising model.
It is believed that:
\begin{itemize}
\item[(A)] For $t < t_c$, $\rho_t(\phi)$ is a continuous function of $\phi$ and positive on all of $\mathbb{T}$.

\item[(B)] For $t \geq t_c$, there is an arc $\mathbb{T} \setminus [-\phi_e(t),\phi_e(t)]$, symmetric about $z=-1$, such that $\rho_t(\phi)$ is a continuous function of $\phi$ and positive
on $\mathbb{T} \setminus [-\phi_e(t),\phi_e(t)]$ and zero otherwise.  Moreover, $\phi_e:[t_c,1]\rightarrow[0,\pi]$ is a continuous function with $\phi_e(t_c)=0$, $\phi_e(t)>0$ for $t>t_c$, and $\phi_e(1)=\pi$.
\end{itemize}
In fact, for sufficiently small $t > 0$, it has been proved by Biskup, Borgs,
Chayes, Kleinwaks, and Koteck\'y \cite{BBCKK} that the limiting measure of
Lee--Yang zeros for the $\mathbb{Z}^d$ lattice is absolutely continuous and
even has $C^2$ density $\rho_t(\phi)$.  Meanwhile, at high temperatures,
quantum field theory gives a prediction of the universal exponents of the
densities $\rho_t$ near the end-points of $\mathbb{T} \setminus
[-\phi_e(t),\phi_e(t)]$, see Fisher \cite{Fis1} and Cardy \cite{Car}. For
example, for $d=2$ the exponent is $(-1/6)$, while for $d>6$ it is $1/2$.
A more detailed discussion of this conjectural behavior for the limiting
measures of Lee--Yang zeros for the $\mathbb{Z}^d$ lattice, including a
discussion of what has been proved, is presented\footnote{Remark that in that
paper, the variable $z=e^{-h/T}$ is used, so the description must be read with
care.} in Section 1 of~\cite{BLR1} and also in the first two sections of
\cite{MBT}.

\subsection{Description of $\mu_t$ for the Diamond Hierarchical Lattice.}
Besides the one-dimensional lattice $\mathbb{Z}^1$ there are very few lattices
for which a global description of the limiting measure of Lee--Yang zeros
has been rigorously proved.  One exception is the Diamond Hierarchical  Lattice
(DHL), which was recently studied in \cite{BLR1}.  Below the critical
temperature of the DHL, the limiting measure of Lee--Yang zeros
matches nicely with the conjectural picture for the $\mathbb{Z}^d$ lattice in
that it is absolutely continuous and even has $C^\infty$ density.  On the other
hand, the sequence of graphs $\Gamma_n$ comprising the DHL has vertices whose
valence tend to infinity, causing the limiting measure of Lee--Yang zeros
$\mu_t$ to have support equal to the entire circle $\mathbb{T}$ for every $t
\in [0,1]$, which is not physical; see \cite{RUELLE_PR}.\\

\subsection{Work of M\"uller-Hartmann--Zittartz, Barata-Marchetti, and Barata-Goldbaum.}
Let $\Gamma_n^k$ denote the $n$-th-level rooted Cayley Tree with
branching number $k$ and $\widehat{\Gamma}_n^{k}$ the unrooted (full)  Cayley Tree of
level $n$ with branching number~$k$.  An illustration for $k=2$ is given in
Figure~\ref{fig:cayleytree}.  We will denote the corresponding lattices by
$\Gamma^k:= \{\Gamma_n^k\}_{n=0}^\infty$ and $\widehat \Gamma^k:= \{\widehat
\Gamma_n^k\}_{n=0}^\infty$.  In Proposition~\ref{PROP:ROOTED_AND_FULL_SAME_MEASURE} we will see
that the limiting measure of Lee--Yang zeros is the same for $\Gamma^k$ and
$\widehat \Gamma^k$ and after that point we will ignore the distinction between
them.

The critical temperature for the Ising Model on the Cayley Tree with branching number $k$ is
\begin{align*}
t_c = \frac{k-1}{k+1}.
\end{align*}
In \cite{MULLERZITTARZ}, M\"uller-Hartmann and Zittartz used the hierarchical structure of the
Cayley Tree to write an explicit expression (see \cite[Eq. 4]{MULLERZITTARZ})
for the limiting free energy. 
They then used this expression to see a curious
type of phase transition: for fixed $0 < t < t_c$ and varying $z \in (0,\infty)$ there exists real-analytic  $F_{\rm reg}(z,t)$                      
so that 
\begin{align}\label{EQN:EXPONENT_FSING}
\lim_{z \rightarrow 1} \frac{\log |F(z,t) - F_{\rm reg}(z,t)|}{\log |1-z|}  = \frac{\log k}{\log \gamma} \qquad \text{where} \qquad \gamma = (B_{1,t,k})'(1) = k \frac{1-t}{1+t}.
\end{align}
Said differently, the singular part of the free energy $F_{\rm sing}(z,t) := F(z,t) - F_{\rm reg}(z,t)$ vanishes with exponent 
\begin{align*}
\kappa(t) := \frac{\log k}{\log \gamma}
\end{align*} 
at $z=1$,
so that $\kappa(t)$ is called the {\em critical exponent} of $F(z,t)$.   The phase transition is called ``continuous order'' because the exponent
$\kappa(t)$ increases continuously from $1$ to $\infty$ as $t$ increases from $0$ to $t_c$.
Meanwhile, for fixed $t_c < t \leq 1$, $F(z,t)$ varies analytically for all $z \in (0,\infty)$.

In \cite{MULLERHARTMAN}, M\"uller-Hartmann provided a different explanation for
this phenomenon by computing the {\em pointwise dimension} $d_{\mu_t}(0)$ of
the Lee--Yang measure $\mu_t$ at $\phi = {\rm Arg}(z) = 0$.  He found\footnote{We have
re-expressed his result in our variables.}:
\begin{align*}
d_{\mu_t}(0) := \lim_{\delta \rightarrow 0}   \frac{\log \mu_t (-\delta,+\delta)}{\log 2\delta} = \kappa(t).
\end{align*}
He then used the electrostatic representation (\ref{electrostat rep}) for
$F(z,t)$ and a clever argument to reprove (\ref{EQN:EXPONENT_FSING}) and
actually obtain further details of the singularity.


A global study of the Lee--Yang zeros for the Cayley Tree is done by Barata and Marchetti \cite{BARATAMARCHETTI}.  While they
allow the coupling constants to be chosen as 0 or $J$, at random, we will simply describe
their results in the deterministic setting.  They proved for the binary Cayley tree $\Gamma_n^2$ that
\begin{enumerate}
\item For $t<t_c = \frac{1}{3}$, the Lee--Yang zeros $\Gamma_n^2$ become dense on $\mathbb{T}$ as $n\rightarrow\infty$.
\item For $t_c \leq t \leq t_\ell \approx 0.46409$, the Lee--Yang zeros of $\Gamma_n^2$ become dense on the arc $\T\setminus[-K(t),K(t)]$ as $n\rightarrow\infty$, where $K:[t_c,t_\ell]\rightarrow[0,\pi]$ is a continuous function such that $K(t_c)=0$, $K(t)>0$ for $t_c < t < t_\ell$, and $K(t_\ell) = \pi$.
\item For $t\geq t_c$, $\Gamma_n^2$ has no Lee--Yang zeros in the arc $[-\phi_e(t),\phi_e(t)]$, where $\phi_e:[t_c,1]\rightarrow[0,\pi]$ is a continuous function such that $\phi_e(t_c)=0$, $\phi_e(t)>0$ for $t>t_c$, and $\phi_e(1)=\pi$.
\end{enumerate}
We refer the reader to Equation \ref{eq:kappa} for the explicit formula of
$\phi_e$ (for arbitrary branching number) and to Figure \ref{img:kvK} for a
plot of the curve formed by $\{(\phi,t) \, : \, \phi = \pm \phi_e(t)\}$.  We refer
the reader to  \cite[Thm. 1.2]{BARATAMARCHETTI} for the explicit formula of
$K$.  (It will not be used in the present paper.) 

\begin{figure}[h!]
{
\centering
\begin{picture}(0,0)%
\includegraphics{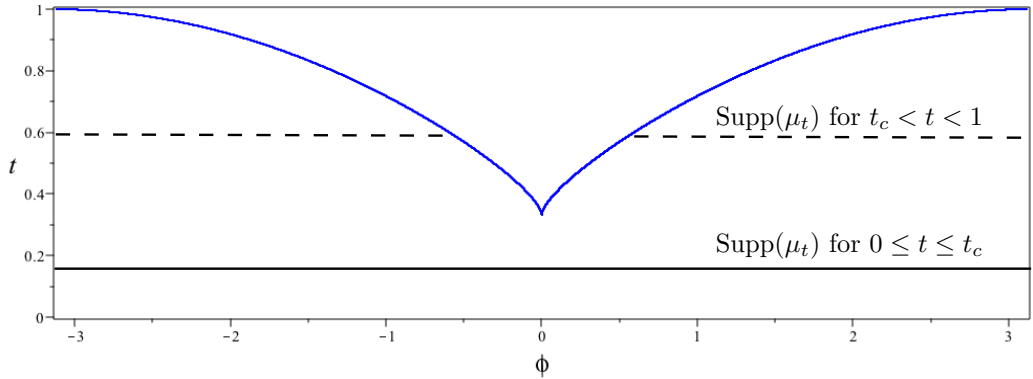}%
\end{picture}%
\setlength{\unitlength}{3947sp}%
\begingroup\makeatletter\ifx\SetFigFont\undefined%
\gdef\SetFigFont#1#2#3#4#5{%
  \reset@font\fontsize{#1}{#2pt}%
  \fontfamily{#3}\fontseries{#4}\fontshape{#5}%
  \selectfont}%
\fi\endgroup%
\begin{picture}(6517,2444)(545,-2805)
\put(5052,-1170){\makebox(0,0)[lb]{\smash{{\SetFigFont{10}{12.0}{\familydefault}{\mddefault}{\updefault}{\color[rgb]{0,0,0}${\rm Supp}(\mu_t)$ for $t_c < t < 1$}%
}}}}
\put(5052,-1972){\makebox(0,0)[lb]{\smash{{\SetFigFont{10}{12.0}{\familydefault}{\mddefault}{\updefault}{\color[rgb]{0,0,0}${\rm Supp}(\mu_t)$ for $0 \leq t \leq t_c$}%
}}}}
\end{picture}%
}
\caption{The curve formed by $\{(\phi,t) \, : \, \phi = \pm \phi_e(t)\}$ for $k=2$.  The support of $\mu_t$ is shown for $0 \leq t \leq t_c$ (solid horizontal line) and for $t_c < t < 1$ (dashed horizontal line).}
\label{img:kvK}
\end{figure}

The work of Barata-Goldbaum \cite{BARATAGOLDBAUM} re-investigates the issues studied by Barata-Marchetti
with the couplings between neighboring vertices chosen periodically and aperiodically.

\subsection{Main Results.}\label{SEC:MAIN_RESULTS}

Each of the following theorems holds for either the rooted Cayley Tree or the
full one.  So, we will not distinguish between them.  Note also that for any
lattice, the limiting measure of Lee--Yang zeros at $t=0$ and $t=1$ are
Lebesgue measure on $\mathbb{T}$ and a Dirac mass at $z=-1$, respectively.
Indeed, for any connected graph $\Gamma = (V,E)$ the Lee--Yang zeros at $t=0$
are the $|V|$-th roots of $-1$ and the Lee--Yang zeros at $t=1$ are all equal to $z=-1$.
For this reason, our results focus on $0 < t < 1$.

In Theorem A we see for the Cayley Tree that $\rm{Supp}(\mu_t)$ matches perfectly with the conjectural picture for the $\mathbb{Z}^d$ lattice:

\begin{THMA}
	Consider the Cayley Tree with branching number $k \geq 2$ and let $\mu_t$ denote the limiting measure of Lee--Yang zeros as $n\rightarrow\infty$ at temperature $t\in[0,1]$. Then
	\begin{itemize}
		\item[{\rm (i)}] for every $0 \leq t \leq t_c$, $\rm{Supp}(\mu_t)=\T$ and
		\item[{\rm (ii)}] For each $t_c < t < 1$, $\rm{Supp}(\mu_t) = \mathbb{T} \setminus (-\phi_e(t),\phi_e(t))$, 
		where
		\begin{equation}
		\begin{split}
		\phi_e(t)&=\text{Arg}\left(w_{\bullet}\left(\frac{1+w_{\bullet}t}{w_{\bullet} + t}\right)^k\right)\text{, where}\\
		w_{\bullet}&=\frac{(k+1)t^2-(k-1) \pm \sqrt{\left( k+1 \right) ^{2}{t}^{4}-2\, \left( {k}^{2}+1 \right) {t}^{2}+
				\left( k-1 \right) ^{2} }}{2t}.
		\end{split}
		\label{eq:kappa}
		\end{equation}
		\item[{\rm (iii)}] For each $t_c < t \leq 1$ and any $n \geq 0$ there are no Lee--Yang zeros for $\Gamma^k_n$ in $(-\phi_e(t),\phi_e(t))$.
	\end{itemize}
\end{THMA}

\noindent
For the remainder of the paper we will use the notation $S_t :=
\rm{Supp}(\mu_t)$, which by Theorem A is either all of $\mathbb{T}$, if $0 \leq
t \leq t_c$, or the arc $\mathbb{T} \setminus (-\phi_e(t),\phi_e(t))$, if $t_c
< t \leq 1$. 

From the formula of $\phi_e$ given in (\ref{eq:kappa}), it is evident that the
set of points $(\pm\phi_e(t),t)$ forms a continuous curve that wraps once
around the cylinder $\T\times[0,1]$. We refer to this as \textit{the $\phi_e$
curve}. It is shown for branching number $k=2$ in Figure \ref{img:kvK}.  

Remark that a generalization of Theorem A to arbitrary graphs with prescribed bounds on the valence of vertices was recently proved by Peters and Regts \cite{PR}.

A key aspect of the proof of Theorem A is the following:
\begin{proposition}\label{PROP:EXPANDING_BP} Fix any branching number $k \geq 2$.  If either $t \in [0,t_c)$ or both $t \in [t_c,1)$ and  $\phi \in \mathbb{T} \setminus [-\phi_e(t),\phi_e(t)]$, then
	\begin{itemize}
		\item[{\rm (i)}] $B_{\phi,t,k}(w)$ has a fixed point $\WD$ in $\D$ and a symmetric fixed point $\WCMD = 1/\overline{\WD} \in\CMD$, and
		\item[{\rm (ii)}] $B_{\phi,t,k}$ is an expanding map of $\mathbb{T}$, i.e.\ there exists $c>0$ and $\lambda>1$ such that for all $w \in \T$ and $n>0$ we have $|\left({B}_{z,t,k}^{n}\right)^{\prime}(w)| \geq c\lambda^n$.
	\end{itemize}
\end{proposition}
\noindent
Throughout the paper we will write $B_{z,t,k}$ and $B_{\phi,t,k}$ interchangably, where $\phi = {\rm Arg}(z)$, and we will omit the branching number $k$ from the notation when it is clear from the context.

In Theorem B we see that for the limiting measure $\mu_t$ for the Cayley Tree is much wilder than what is conjectured (and rigorously proved at small temperature \cite{BBCKK}) for the $\mathbb{Z}^d$ lattice:

\begin{THMB} Fix any branching number $k \geq 2$ and any $0 < t < 1$.
For any compact interval $X_t \subseteq {\rm interior}(S_t)$, the restriction of $\mu_t$ to $X_t$ has Hausdorff dimension less than one. In particular, for any $\phi \in S_t$ and any neighborhood of $\phi$, $\mu_t$ is not absolutely continuous with respect to the Lebesgue measure.
\end{THMB}

\noindent
(Recall that the Hausdorff dimension of a measure is defined to be the smallest Hausdorff dimension of a full measure set.)

\vspace{0.1in}

The pointwise dimension $d_{\mu_t}(\phi)$ for $\mu_t$ at $\phi \in S_t$ is defined by
\begin{align*}
d_{\mu_t}(\phi) :=  \lim_{\delta \rightarrow 0}   \frac{\log \mu_t ([\phi-\delta,\phi+\delta])}{\log 2\delta},
\end{align*}
supposing that the limit exists.

\begin{THMC}\label{THMC}
Fix any branching number $k \geq 2$.
For any $0 < t < 1$, there is a Lebesgue full measure set $S^+_t \subset S_t$, such that for any $\phi \in S^+_t$, we have
\begin{equation}
d_{\mu_t}(\phi)=\frac{\log k}{\chi_{\phi,t}} > 1,
\end{equation}
where 
\begin{align}
\chi_{\phi,t} = 2\pi\log \left|\frac{k(1-t^2)w_{\mathbb{D}}(1-w_{\mathbb{D}}t)}
         {(w_{\mathbb{D}}+t) (1+w_{\mathbb{D}}t)(t-w_{\mathbb{D}}) }\right|,
\end{align}
with $w_{\mathbb{D}}$ the unique fixed point of $B_{z,t,k}$ in $\D$.

\vspace{0.1in}
\noindent
Meanwhile, there is a dense set $S^-_t \subset S_t$, such that for any $\phi \in S^-_t$ we have $d_{\mu_t}(\phi) < 1$.
\end{THMC}

\noindent
(We will see in the proof of Theorem C that $\chi_{\phi,t}$ is the Lyapunov exponent for the unique absolutely continuous invariant measure $\nu_{z,t}$ for $B_{z,t,k}$.)

Theorem C and an adaptation of M\"uller-Hartmann's analysis of the electrostatic
representation~(\ref{electrostat rep}) for $F(z,t)$ allows us to prove the
following global (Lebesgue almost everywhere) version of the ``phase transition of
continuous order'' described by  M\"uller-Hartmann and Zittartz.

For $\phi \in S_t^{+}$ we will see that the pointwise dimension
$d_{\mu_t}(\phi)$ serves as the critical exponent for the free energy, thus it
will be more natural to use the notation $\kappa(\phi,t) \equiv
d_{\mu_t}(\phi)$ in order to be consistent with the notations from \cite{MULLERZITTARZ,MULLERHARTMAN}.

\begin{THMD}
Fix any branching number $k \geq 2$, any $0 < t < 1$, and any $\phi$
in the Lebesgue full measure set $S^{+}_t\subset S_t$.   Then, the free energy $F(z,t)$ has radial critical
exponent
\begin{align*}
\kappa(\phi,t) = d_{\mu_t}(\phi)=\frac{\log k}{\chi_{\phi,t}}
\end{align*}
at the point $z=e^{i \phi}$.  More precisely, 
there is a real analytic function\footnote{$g$ depends on on $k, t,$ and $\phi$.} $g:(0,\infty) \rightarrow \mathbb{R}$ so that
\begin{align*}
\lim_{r\rightarrow 1} \frac{\log \left|F\left(r e^{i\phi} ,t\right) - g(r)\right|}{\log |r-1|} = \kappa(\phi,t) > 1.
\end{align*}
Meanwhile, for $\phi$ in the dense set $S^-_t$ there is a real analytic $g:(0,\infty) \rightarrow \mathbb{R}$ so that
\begin{align*}
\lim_{r\rightarrow 1} \frac{\log \left|F\left(r e^{i\phi},t\right) - g(r)\right|}{\log |r-1|} < 1.
\end{align*}
\end{THMD}

\noindent
The ``almost-everywhere'' critical exponent $\kappa(\phi,t)$ from Theorem D is illustrated in Figure \ref{FIG:CRITICAL_EXPONENT}.

\begin{figure}
\begin{picture}(0,0)%
\includegraphics{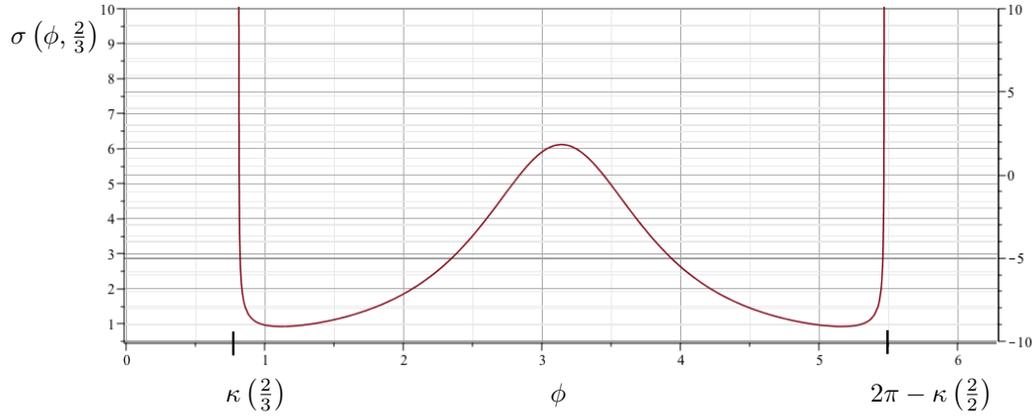}%
\end{picture}%
\setlength{\unitlength}{3947sp}%
\begingroup\makeatletter\ifx\SetFigFont\undefined%
\gdef\SetFigFont#1#2#3#4#5{%
  \reset@font\fontsize{#1}{#2pt}%
  \fontfamily{#3}\fontseries{#4}\fontshape{#5}%
  \selectfont}%
\fi\endgroup%
\begin{picture}(6465,2611)(1936,-2897)
\put(7351,-2833){\makebox(0,0)[lb]{\smash{{\SetFigFont{10}{12.0}{\familydefault}{\mddefault}{\updefault}{\color[rgb]{0,0,0}$2\pi - \kappa\left(\frac{2}{2}\right)$}%
}}}}
\put(1951,-586){\makebox(0,0)[lb]{\smash{{\SetFigFont{10}{12.0}{\familydefault}{\mddefault}{\updefault}{\color[rgb]{0,0,0}$\sigma\left(\phi,\frac{2}{3}\right)$}%
}}}}
\put(3301,-2833){\makebox(0,0)[lb]{\smash{{\SetFigFont{10}{12.0}{\familydefault}{\mddefault}{\updefault}{\color[rgb]{0,0,0}$\kappa\left(\frac{2}{3}\right)$}%
}}}}
\put(5341,-2833){\makebox(0,0)[lb]{\smash{{\SetFigFont{10}{12.0}{\familydefault}{\mddefault}{\updefault}{\color[rgb]{0,0,0}$\phi$}%
}}}}
\end{picture}%
\caption{Plot of the ``almost everywhere critical exponent'' $\kappa\left(\phi,\frac{2}{3}\right)$ for branching number $k=2$. \label{FIG:CRITICAL_EXPONENT}}
\end{figure}

\subsection{Main technical difficulty}
Let us describe the main technical difficulty in proving Theorems~B and C.
It begins with the fact that Proposition
\ref{PROP:RENORM} expresses the Lee--Yang zeros for $\Gamma_n$ as solutions to
$B^{n}_{z,t}(z) = -1$, i.e.\ that variable $z$ occurs both as a parameter and
as the dynamical variable.  To address this issue we fix $t \in (0,1)$ and work
with the skew product
\begin{align*}
B: \mathbb{T} \times \mathbb{T} \rightarrow \mathbb{T} \times \mathbb{T} \qquad \mbox{given by} \qquad B(\phi,\theta) = (\phi,B_{\phi,t}(\theta)),
\end{align*}
where\footnote{We will typically abuse notation and ignore subtleties about
branches ${\rm arg}$, except when truly necessary.}  we have set $\phi = {\rm
arg}(z)$ and $\theta = {\rm arg}(w)$.  Suppose we parameterize the diagonal
$\Delta = \{(\phi,\theta) \, : \, \theta = \phi\}$ by the variable $\phi$.
Then, Proposition \ref{PROP:RENORM} gives that the Lee--Yang zeros for
$\Gamma_n$ are the intersection points
\begin{align*}
(B^{n})^{-1}\{\theta = \pi\} \, \cap \, \Delta.
\end{align*}
This is illustrated in Figure \ref{FIG:SKEW_PRODUCT_ILLUSTRATION}.

\begin{figure}
\includegraphics[scale=0.5]{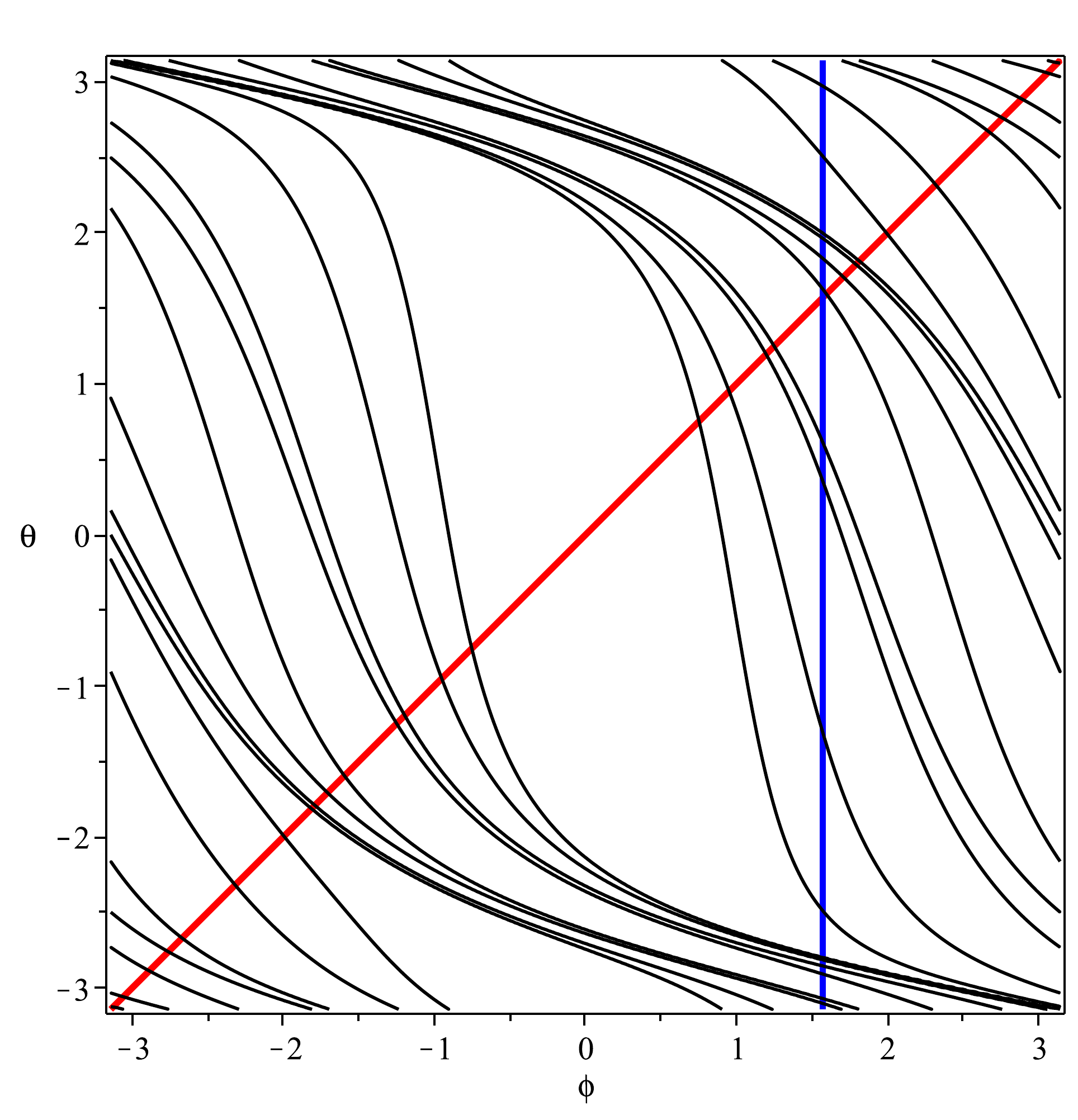}
\caption[LoF entry]{\label{FIG:SKEW_PRODUCT_ILLUSTRATION}  Illustration of how to
determine Lee--Yang zeros for $\Gamma^k_n$ using the skew product
$B(\phi,\theta) = (\phi,B_{\phi,t}(\theta))$. 
Here, we use branching number $k=2$, level $n=4$, and temperature $t=\frac{1}{2}$.  For these choices,
the Lee--Yang zeros are the values of $\phi$ at which the black curves (depicting $B^{-4}\{\theta = \pi\}$) intersect the diagonal $\Delta$, in red.
Also shown is a vertical $\mathbb{T}_{\frac{\pi}{2}} := \{\frac{\pi}{2}\} \times \mathbb{T}$, 
in blue.}
\end{figure}

Fix $\phi_0 \in {\rm interior}(S_t)$ and let $\mathbb{T}_{\phi_0} :=
\{\phi_0\} \times \mathbb{T}$.  Then, 
$B_{\phi_0,t}: \mathbb{T}_{\phi_0} \rightarrow \mathbb{T}_{\phi_0}$ is an
expanding map of the circle, by Proposition \ref{PROP:EXPANDING_BP}.   Therefore, if we assign Dirac mass to each of the
preimages $(B^{-n})\{\theta = \pi\} \, \cap \,
\mathbb{T}_{\phi_0}$ and normalize, the result converges to the measure of
maximal entropy (MME) $\eta_{\phi_0,t}$ of $B_{\phi_0,t}$, as $n\rightarrow \infty$.  Because expanding maps of the circle are one of the simplest types
of dynamical systems, a tremendous amount is known about $\eta_{\phi_0,t}$.  The issue in proving Theorems B and C is to relate
these properties of $\eta_{\phi_0,t}$ to the properties of the Lee--Yang measure $\mu_t$ at points on the diagonal $\Delta$ that are near to $(\phi_0,\phi_0)$.

This is done as follows: Let $X_t \Subset
{\rm interior}(S_t)$ be an interval containing $\phi_0$.  Then,
Proposition~\ref{PROP:EXPANDING_BP} implies that the restriction $B: X_t \times
\mathbb{T} \rightarrow X_t \times \mathbb{T}$ is partially hyperbolic,
with the vertical direction expanding.  As such, it has a unique central
foliation $\mathcal{F}^c$, that can be thought of as ``horizontal''.  Using
standard dynamical techniques, one can construct a holonomy invariant transverse
measure $\eta$ on $\mathcal{F}^c$ that describes the limit as $n \rightarrow
\infty$  of the (normalized) preimages 
\begin{align*}
(B^{n})^{-1}\{\theta = \pi\}.
\end{align*}
If we restrict $\eta$ to $\mathbb{T}_{\phi_0}$ we obtain the MME
$\eta_{\phi_0,t}$ for $B_{\phi_0,t}$ and if we restrict $\eta$ to the diagonal
$\Delta$ we obtain the Lee--Yang measure $\mu_t$.  Therefore, they are related
by holonomy along $\mathcal{F}^c$.

However, the main technical issue now arises because honolomies along
$\mathcal{F}^c$ are not absolutely continuous\footnote{See also \cite{milnor1}
for a similar situation but with a more concrete construction.} so that it is
impossible to control holonomy images of arbitrary sets of zero Lebesgue
measure.  However, these holonomies are H\"older continuous, with H\"older
exponent arbitrarily close to one, so long as the two transversals are chosen
sufficiently close.  This allows us to control the holonomy images of sets
whose Hausdorff dimension is less than one.  So, the key idea in proving
Theorems B and C is to work with sets of Hausdorff Dimension less than one,
instead of working with arbitrary  sets of zero Lebesgue measure.

This idea goes back to conversations the last author had with Victor
Kleptsyn at the
conference\footnote{\url{http://www.math.stonybrook.edu/dennisfest/}} in honor
of Dennis Sullivan's 70th birthday, who had used the same idea in work with Ilyashenko and Saltykov on intermingled basins of attraction \cite{IKS}.  Moreover, a key tool used in proving
Theorem~C is the ``Special Ergodic Theorem'' proved by Kleptsyn, Ryzhov, and Minkov
\cite{Kleptsyn} which allows one to conclude that set of initial conditions
whose ergodic averages deviate by more than $\epsilon > 0$ from the space
average has Hausdorff dimension less than one.  It is a generalization of a preliminary version that appeared in \cite{IKS}.

\subsection{Connection with dynamics of Blaschke Products and  expanding maps of the circle.}

The proof of each of the theorems above rely upon techniques from real and
complex dynamics to study the iterates of the Blaschke product $\bzt$.  The dynamics of Blaschke Products and, more generally, of $C^2$ expanding maps of the circle
is a classical topic in the dynamical systems community \cite{S,PS,PUJALS} and their study remains an active area of
research in dynamics; see \cite{Erchenko,GP,MCMULLEN,IVRII} for a sample.  Remark also that Blaschke Products arise in a far  more
subtle way than here, when studying the Lee--Yang zeros for the Diamond
Hierarchical Lattice \cite{BLR1}.

\subsection{Plan for the paper}

Because the paper is written for readers from both mathematical physics and
dynamical systems, we have made an effort to provide ample details throughout.
Section~\ref{SEC:BASIC_C_DYN} is devoted to proving Proposition
\ref{PROP:EXPANDING_BP}, which is the key statement needed to begin applying dynamical
systems techniques to the problem.  This is followed by a proof of Theorem A in
Section~\ref{SEC:PROOF_THMA}.  In Section~\ref{SEC:BASIC_REAL} we summarize
several powerful results from real dynamics that will be applied to the
expanding Blaschke Products and explain their consequences for $\bzt$.  Section~\ref{SEC:PH} is devoted to studying the skew-product $B(\phi,\theta) =
(\phi,B_{\phi,t}(\theta))$ as a partially hyperbolic mapping and relating the
limiting measure of Lee--Yang zeros to the measure of maximal entropy for
the Blaschke Products under holonomy along the central foliation.  We prove
Theorem B in Section~\ref{SEC_PF_THMB} and then prove Theorem~C in Section~\ref{SEC:PF_THMC}.  Section~\ref{SEC:THMD} is devoted to Theorem D.
Proofs of the formula for Lee--Yang zeros on the Cayley Tree in terms of iterates of $\bzt$ (Proposition \ref{PROP:RENORM}) can be found in the literature, however, we 
provide a derivation in Appendix \ref{SEC:APPENDIX_DERIVATION}, so that our paper can be read independently.

\subsection{Basic Notations Used}
We will use the following notation: the complex open unit disk $\D:=\{z\mid|z|<1\}$ and the Riemann sphere $\C:=\mathbb{C}\cup\{\infty\}$.

\vspace{0.1in}

\noindent\textbf{Acknowledgments}: We thank Pavel Bleher for suggesting this
problem to us and for his several helpful comments.  Many other people have
also provided helpful comments and advice, including Joseph Cima, Vaughn
Climenhaga, Oleg Ivrii, Benjamin Jaye, Victor Kleptsyn, Micha\l \ Misiurewicz,
Boris Mityagin, Juan Rivera-Letelier, William Ross, and Maxim Yattselev.  We
thank the referees for their very careful reading of our paper and their
helpful comments.  Theorem A and its proof were the main content of Anthony Ji
and Caleb He's entry in the 2016 Siemens Competition in Math, Science, and
Technology.  The work of the first and fourth authors was supported by NSF
grant DMS-1348589.


\section{Basic Complex Dynamics for $B_{z,t}$}
\label{SEC:BASIC_C_DYN}
This section is devoted to proving Proposition \ref{PROP:EXPANDING_BP}.   The results in this section could be obtained more quickly by appealing to the results of \cite{PUJALS}; however, we give complete details here.

For $t\in[0,1)$, the map $\bzt$ is an example of a {\em{Blaschke Product}}, which is a map of the form ${\mathcal B}(w)=e^{i\phi}\prod_{i=1}^k \frac{w - a_i}{1-\overline{a_i}w}$, where $\phi$ is real and the $a_i$'s are in $\D$. Blaschke Products satisfy:
\begin{enumerate}
\item Each of $\mathbb{D}, \mathbb{T},\text{and } \mathbb{\hat{C}}\setminus\overline{\mathbb{D}}$ is totally invariant under ${\mathcal B}$. That is, if $S$ is any of the three sets above, then $w \in S$ if and only if ${\mathcal B}(w)\in S$.
\item Blaschke Products are symmetric across $\mathbb{T}$. That is, ${\mathcal B}\left(\frac{1}{\bar{w}}\right) = \frac{1}{\overline{{\mathcal B}(w)}}$ for all $w\in\mathbb{\hat{C}}.$
\end{enumerate}

\begin{lemma}\label{LEM:ATTRACTING}
If the map $\bzt$ has a fixed point $\WD$ in $\D$, then $\WD$ must be attracting, and the orbit of each point in $\D$ must converge to $\WD$. In particular, $\WD$ must be the only fixed point in $\D$. The analogous result follows for a fixed point $\WCMD$ in $\CMD$.
\end{lemma}
\begin{proof}
An easy adaptation of the classical Schwarz Lemma states that if $f: \mathbb{D} \rightarrow \mathbb{D}$ is a holomorphic map with fixed point $f(u) = u$ for some $u\in\mathbb{D}$, then  $|f^{\prime}(u)| \le 1$. Moreover, $|f^{\prime}(u)| = 1$ if and only if $f$ is an automorphism of $\mathbb{D}$ (that is, 1-1 and onto). Meanwhile, if $|f^{\prime}(u)|<1$, then for every $w\in\D$, we have $f^{n}(w)\rightarrow u$ as $n\rightarrow\infty$.

In our setting, $\bzt$ is a degree $k\geq2$ rational map for which $\D$ is
totally invariant, so its restriction to $\mathbb{D}$ cannot be an automorphism.
The result for $\WD$ follows.  
Meanwhile, the result for $\WCMD$ follows by symmetry of the Blaschke Product across $\mathbb{T}$.
\end{proof}

\begin{proposition}\label{PROP:FIXED_POINT_MULTIPLICITY}
$\bzt(w)$ has a fixed point $w_{\bullet}$ of multiplicity greater than 1 if and only if $(\phi,t)$ is on the $\phi_e$ curve.
\end{proposition}
\begin{proof}
Let $\mathcal{P}_{z,t}(w):=z(w+t)^k-w(1+wt)^k$. Then a point $w_\bullet$ is a fixed point of $\bzt$ and has multiplicity greater than 1 if and only if $\pzt(w_\bullet)=0$ and $\pzt^{\prime}(w_\bullet)=0$. Solving these two equations, we find that $w_\bullet$ must satisfy 
$$kw_{\bullet}\left(\frac{1-t^2}{(w_\bullet+t)(1+w_\bullet t)}\right) = 1.$$ Solving yields the two solutions stated in the Theorem A: 
$$w_{\bullet}=\frac{(k+1)t^2-(k-1) \pm \sqrt{\left( k+1 \right) ^{2}{t}^{4}-2\, \left( {k}^{2}+1 \right) {t}^{2}+
 \left( k-1 \right) ^{2} }}{2t}.$$ 
We can use the condition that
$w_{\bullet} = z\left(\frac{w_{\bullet}+t}{1+w_{\bullet}t}\right)^k$ to find that 
$z = w_{\bullet}\left(\frac{1+w_{\bullet}t}{w_{\bullet} + t}\right)^k = e^{i\phi_e(t)}$  if and only if $\bzt$ has a fixed point of multiplicity greater than one.
  \end{proof}

\vspace{0.1in}

\begin{proof}[Proof of Proposition \ref{PROP:EXPANDING_BP}, Part {\rm (i)}:]
By symmetry of the Blaschke Product $\bzt$ across $\mathbb{T}$, if $\bzt$ has a
fixed point in $\mathbb{D}$ then it also has one in $\hat{\mathbb{C}} \setminus
\overline{\mathbb{D}}$ and the two fixed points have the same argument.  By
Lemma \ref{LEM:ATTRACTING}, both such fixed points would be attracting and they
would be the unique fixed points of $\bzt$ in $\mathbb{D}$ and
$\hat{\mathbb{C}} \setminus \overline{\mathbb{D}}$, respectively.

Starting with
initial parameter $(\phi,t) = (0,0)$ yields fixed points at $0$ and $\infty$. Let us
assume there is some $(\phi_0,t_0)$ below the $\phi_e$ curve such that all
fixed points lie on the circle. Since $\bzt$ is a rational function with
complex coefficients, the fixed points vary continuously with its parameters.
We can pick a continuous path from $(0,0)$ to $(\phi_0,t_0)$ that lies below
the $\phi_e$ curve. However, the only way all fixed points can be on the circle
is if the fixed points off the circle collided on the circle, since the fixed
points off the circle have the same argument. This implies the path intersects
the $\phi_e$ curve, a contradiction.  \end{proof}

\begin{proposition}
For $\phit$ below the $\phi_e$ curve, the Julia set of $\bzt$ is $\T$.
\label{prop:julia}
\end{proposition}
\begin{proof}
We proceed by proving that the Fatou set of $\bzt$ is $\C\setminus\T$. We will first show that for any open neighborhood $U\subset\C\setminus\T$, the family of iterates $\bzt^{1}(U), \bzt^{2}(U), \dots$ is \textit{normal}. That is, for any infinite subsequence of the family, there exists a further infinite sub-subsequence that converges to a holomorphic map on $U$.			

Suppose $U\subset\D$, and let $w_\D$ be the fixed point of $\bzt$ that is in $\D$. By Lemma \ref{LEM:ATTRACTING}, it follows that $\bzt^{n}(w)\rightarrow w_\D$ for any $w\in\D$. Thus, the restriction of $\bzt$ to any such $U\subset D$ produces a normal family. Using the symmetry of Blaschke Products, the analogous result follows for $U\subset\CMD$.

Suppose there exists a point $w\in\T$ and an open neighborhood $U$ around $w$ on which the iterates of $\bzt$ form a normal family. As already proved, we have $\bzt^{n}(U\cap\D)\rightarrow w_\D$, while $\bzt^{n}(U\cap\CMD)\rightarrow w_{\CMD}$. Thus, the family produced by $U$ converges to a discontinuous function and therefore not holomorphic, contradiction. Thus, no point in $\T$ is in the Fatou set, so $\T$ is the Julia set.
\end{proof}

\begin{proof}[Proof of Proposition \ref{PROP:EXPANDING_BP}, Part {\rm (ii)}:]
The following classical theorem that can be found in \cite[Theorem 14.1]{MILNOR}.
\begin{expansion}
For a rational map $f : \hat{\mathbb{C}}\rightarrow\hat{\mathbb{C}}$ of degree $d\geq 2$ , the following
two conditions are equivalent:
\begin{itemize}
\item[{\rm (i)}] $f$ is expanding on its Julia set $J$ (there exists $c>0$ and $\lambda>1$ such that $|\left(f^{n}\right)^{\prime}(w)| \geq c\lambda^n$ for all $w\in J$).
\item[{\rm (ii)}] The forward orbit of each critical point of $f$ converges towards some
attracting periodic orbit.
\end{itemize}
\end{expansion}

The Julia set of $\bzt$ is $\T$, by Proposition \ref{prop:julia}, so it
suffices to check that (2) holds.  By Part (i) of Proposition
\ref{PROP:EXPANDING_BP}, $\bzt$ has fixed points $\WD\in\D$ and $\WCMD\in\CMD$.
The only critical points of $\bzt$ are at $w=-t$ and $w=-\frac{1}{t}$, which
are in $\D$ and $\CMD$, respectively, since $t \in [0,1)$. By Lemma
\ref{LEM:ATTRACTING}, the orbits of these critical points converge to $\WD$ and
$\WCMD$, respectively.  
\end{proof}

\section{Proof of Theorem A}
\label{SEC:PROOF_THMA}

Throughout the rest of the paper we will focus on the dynamics of $\bzt: \mathbb{T} \rightarrow \mathbb{T}$.
Let us write the mapping in angular form in terms of  $\theta = {\rm Arg}(w)$ and $\phi = {\rm Arg}(z)$.  While it will be sufficient to consider $\phi \in [-\pi,\pi]$,
it will be helpful to allow $\theta \in \mathbb{R}$, so we consider the angular form of $\bzt$ as a ``lift'' to $\mathbb{R}$:
\begin{equation}
\btilde_{\phi,t,k} : \mathbb{R} \rightarrow \mathbb{R} \qquad \mbox{where} \qquad \btilde_{\phi,t,k}(\theta) = k\theta-2k\arctan\left(\frac{t\sin\theta}{1+t\cos\theta}\right) + \phi.
\label{eq:angular}
\end{equation}
\noindent
Note that for any $\phi,t,k$ and $\theta$ the lift satisfies
\begin{align}\label{EQN:DEGK}
\btilde_{\phi,t,k}(\theta + 2\pi) = \btilde_{\phi,t,k}(\theta) + 2\pi k
\end{align}
reflecting the fact that $B_{z,t,k} : \mathbb{T} \rightarrow \mathbb{T}$ is a degree $k$ mapping of the circle.

As usual, we will drop the parameter $k$ when it is clear from the context, writing $\btilde_{\phi,t} \equiv \btilde_{\phi,t,k}$.

\begin{remark}
\label{rmk:restate}
We can restate Proposition $\ref{lyz}$ as the following: \
\begin{itemize}
\item[{\rm (i)}] $e^{i\phi}$ is a Lee--Yang zero of $\Gamma_n^k$ if and only if $\btilde_{\phi,t,k}^{n}(\phi) \equiv \pi\mod 2\pi$.
\item[{\rm (ii)}] $e^{i\phi}$ is a Lee--Yang zero of $\widehat{\Gamma}_n^k$ if and only if $\btilde_{\phi,t,k+1} \circ \btilde_{\phi,t,k}^{n-1}(\phi) \equiv \pi\mod 2\pi$.
\end{itemize}
\end{remark}
\begin{remark}
For every $n>0$, the modulus of the derivative of $\bphittilde^{n}(\theta)$ with respect to $\theta$ coincides with that of $\bzt^{n}(w)$ with respect to $w$. 
In particular, if $(\phi,t)$ is below the $\phi_e$ curve, there exists $c>0$ and $\lambda>1$ such that $(\bphittilde^{n})^\prime(\theta) \geq c\lambda^n$ for all $\theta\in
\mathbb{R}$ and $n>0$.
\label{rmk:expansionforA}
\end{remark}

\begin{lemma}
If $\phi_1\leq\phi_2$ and $\theta_1\leq\theta_2$, then $\btilde_{\phi_1,t}^{n}(\theta_1)\leq \btilde_{\phi_2,t}^{n}(\theta_2)$ for any $n\in\mathbb{N}$.
\label{lm:mono}
\end{lemma}
\begin{proof}
Note $\btilde_{\phi,t}(\theta)$ increases in $\phi$ trivially. Moreover,
\begin{equation}
\frac{d \btilde_{\phi,t,k}}{d\theta} = k\left(\frac{1-t^2}{1+2t\cos\theta + t^2}\right) > 0,
\end{equation}
so $\bphittilde(\theta)$ increases in $\theta$ as well. By induction, the assertion can be proved for $n>1$.
\end{proof}

\subsection{Proof of Parts (i) and (ii)}
We will focus on the proof for the rooted Cayley Tree.  Then, at the end of the subsection, explain how to adapt it for the unrooted version.

We will show that $\mu_{t}((\phi_1,\phi_2)) > 0$ for any interval $(\phi_1,\phi_2)$ not intersecting $(-\phi_e(t),\phi_e(t))$.
Because the Lee--Yang zeros are symmetric about $\phi = 0$, we can suppose that $0 < \phi_1 < \phi_2$.
The limiting measure $\mu_t$ is given by
\begin{align*}
\mu_{t}((\phi_1,\phi_2)) &= \lim_{n\rightarrow\infty}\frac{\# \text{ of LY zeros in } (\phi_1,\phi_2) \text{ for }\Gamma_n^k}{\text{total }\# \text{ of LY zeros for }\Gamma_n^k}\\ 
&= \lim_{n\rightarrow\infty}\frac{\frac{1}{2\pi}(\btilde^{n}_{\phi_2,t}(\phi_2)-\btilde^{n}_{\phi_1,t}(\phi_1))}{(k^{n+1}-1)\cdot (k-1)^{-1}}.
\end{align*}
The second equality holds since the numerators in the two limits differ by at most
1 for every $n>$~$0$; see Remark \ref{rmk:restate}. Therefore, it suffices to show that 
$\btilde^{n}_{\phi_2,t}(\phi_2)-\btilde^{n}_{\phi_1,t}(\phi_1)$
grows exponentially at rate $k$.

Since $(\phi_1,t)$ is below the $\phi_e$ curve, Part (ii) of Proposition \ref{PROP:EXPANDING_BP} (see also Remark \ref{rmk:expansionforA}) gives constants
$c>0$ and $\lambda>1$ such that
\begin{align*}
\btilde^{n}_{\phi_1,t}(\phi_2) - \btilde^{n}_{\phi_1,t}(\phi_1) &= \int_{\phi_1}^{\phi_2} \left(\btilde^{n}_{\phi_1,t}\right)^\prime
(\theta) d\theta 
\geq (\phi_2-\phi_1)c\lambda^n.
\end{align*}
In particular, there exists some $N > 0$ such that $\btilde^{N}_{\phi_1,t}(\phi_2) - \btilde^{N}_{\phi_1,t}(\phi_1) > 2\pi$.  Then, for any 
$n \geq N$ we have
\begin{align}\label{EQN:EXPGROWTH}
\btilde^{n}_{\phi_2,t}(\phi_2)-\btilde^{n}_{\phi_1,t}(\phi_1) &\geq 
\btilde^{n}_{\phi_1,t}(\phi_2)-\btilde^{n}_{\phi_1,t}(\phi_1)  \nonumber \\
&\geq \btilde^{n-N}_{\phi_1,t}\left(\btilde^N(\phi_1) + 2\pi\right)-\btilde^{n-N}_{\phi_1,t}\left(\btilde^N(\phi_1)\right) = 2\pi k^{n-N},
\end{align}
with the first two inequalities given by Lemma \ref{lm:mono} and the last equality coming from the fact that $\bzt$ has degree $k$; see Equation (\ref{EQN:DEGK}).
Therefore, $\mu_{t}((\phi_1,\phi_2)) > 0$.

\vspace{0.1in}
In the case of the unrooted Cayley Tree we have 
\begin{align*}
\widehat{\mu}_{t}((\phi_1,\phi_2)) &= \lim_{n\rightarrow\infty}\frac{\# \text{ of LY zeros in } (\phi_1,\phi_2) \text{ for }\widehat{\Gamma}_n^k}{\text{total }\# \text{ of LY zeros for }\widehat \Gamma_n^k}\\
&= \lim_{n\rightarrow\infty}\frac{\frac{1}{2\pi}\left(\btilde_{\phi_2,t,k+1} \circ \btilde^{\circ n-1}_{\phi_2,t,k}(\phi_2)-\btilde_{\phi_1,t,k+1} \circ \btilde^{n-1}_{\phi_1,t,k}(\phi_1)\right)}{(k^{n+1}+k-2)\cdot (k-1)^{-1}}.
\end{align*}
As for the rooted tree, we need to show that the numerator grows exponentially at rate $k$.  Again, using Lemma \ref{lm:mono}, it suffices to
prove it for the smaller quantity
\begin{align*}
\btilde_{\phi_1,t,k+1} \circ \btilde^{n-1}_{\phi_1,t,k}(\phi_2)-\btilde_{\phi_1,t,k+1} \circ \btilde^{n-1}_{\phi_1,t,k}(\phi_1).
\end{align*}
However, this follows from Equation (\ref{EQN:EXPGROWTH}) and the fact that there is a uniform constant $A$ such that
\begin{align*}
\frac{d \btilde_{\phi,t,k+1}}{d\theta} = (k+1)\left(\frac{1-t^2}{1+2t\cos\theta + t^2}\right) > A > 0.
\end{align*}
\noindent
\qed (Theorem A, parts (i) and (ii).)

\subsection{Proof of Part (iii) of Theorem A}
For the rooted Cayley Tree, Part (iii) of Theorem A was proved by Barata-Marchetti \cite{BARATAMARCHETTI}.  We include their proof here for completeness and we also explain the
adaptation needed for the full Cayley Tree.

Let us start with the rooted tree.  Since the Lee--Yang zeros are symmetric under $\phi \mapsto -\phi$, it suffices to prove that 
for any $t>t_c$ there are no Lee--Yang zeros for $\Gamma_n$ at any angle $\phi_0 \in [0,\phi_e(t))$.
It suffices to show that for any such $t$, $\phi_0$, and $n$ we have that $\btilde_{\phi_0,t}^{n}\left(\phi_0\right)<~\pi$.

The proof is illustrated by Figure \ref{proof2}. Here, the graph of $\btilde_{\phi_0,t}(\theta)$ is depicted by the blue curve and the diagonal is depicted by the red line. By the definition of $\phi_e$, parameter $\phi_0=\phi_e$ refers to the case when the blue curve lies tangent to the red curve at $\theta_\bullet<\pi$, where $\theta_\bullet = \text{Arg}(w_\bullet)$ with $w_\bullet$ given in (\ref{eq:kappa}). However, since $0<\phi_0<\phi_e (t)$ and $\btilde_{\phi_0,t}(\theta)$ is increasing and concave up for $\theta\in(0,\pi)$, there must exist some $\theta_*$ with $\theta_*<\theta_\bullet<\pi$ such that the iterates (black staircase) of $\phi_0$ under $\btilde_{\phi_0,t}$ converge to $\theta_*$, as shown in the figure.
\begin{figure}[!h]
	\begin{center}
		\scalebox{0.9}{
\begin{picture}(0,0)%
\includegraphics{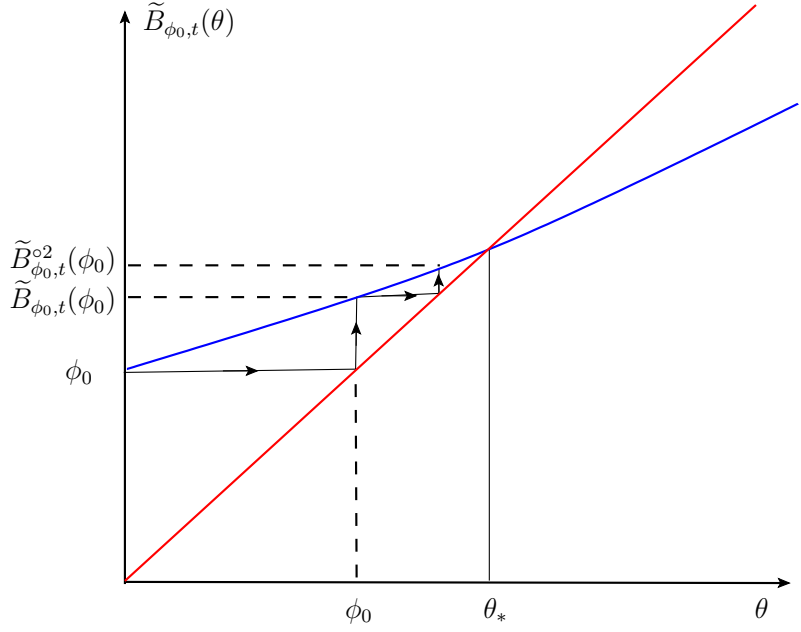}%
\end{picture}%
\setlength{\unitlength}{3947sp}%
\begingroup\makeatletter\ifx\SetFigFont\undefined%
\gdef\SetFigFont#1#2#3#4#5{%
  \reset@font\fontsize{#1}{#2pt}%
  \fontfamily{#3}\fontseries{#4}\fontshape{#5}%
  \selectfont}%
\fi\endgroup%
\begin{picture}(5540,4378)(3346,-5539)
\put(3361,-3301){\makebox(0,0)[lb]{\smash{{\SetFigFont{12}{14.4}{\familydefault}{\mddefault}{\updefault}{\color[rgb]{0,0,0}$\btilde_{\phi_0,t}(\phi_0)$}%
}}}}
\put(4291,-1359){\makebox(0,0)[lb]{\smash{{\SetFigFont{12}{14.4}{\familydefault}{\mddefault}{\updefault}{\color[rgb]{0,0,0}$\btilde_{\phi_0,t}(\theta)$}%
}}}}
\put(5702,-5461){\makebox(0,0)[lb]{\smash{{\SetFigFont{12}{14.4}{\familydefault}{\mddefault}{\updefault}{\color[rgb]{0,0,0}$\phi_0$}%
}}}}
\put(6669,-5461){\makebox(0,0)[lb]{\smash{{\SetFigFont{12}{14.4}{\familydefault}{\mddefault}{\updefault}{\color[rgb]{0,0,0}$\theta_*$}%
}}}}
\put(8562,-5461){\makebox(0,0)[lb]{\smash{{\SetFigFont{12}{14.4}{\familydefault}{\mddefault}{\updefault}{\color[rgb]{0,0,0}$\theta$}%
}}}}
\put(3750,-3797){\makebox(0,0)[lb]{\smash{{\SetFigFont{12}{14.4}{\familydefault}{\mddefault}{\updefault}{\color[rgb]{0,0,0}$\phi_0$}%
}}}}
\put(3361,-3019){\makebox(0,0)[lb]{\smash{{\SetFigFont{12}{14.4}{\familydefault}{\mddefault}{\updefault}{\color[rgb]{0,0,0}$\btilde^{\circ 2}_{\phi_0,t}(\phi_0)$}%
}}}}
\end{picture}%
		}
	\end{center}
	\caption{Iterates (black staircase) of $\phi_0$ under $\btilde_{\phi_0,t}$ that approach $\theta_*$ for $0~<~\phi_0~<~\phi_e(t)$. Blue depicts the graph of $\btilde_{\phi_0,t}(\theta)$ and the diagonal is depicted in red.}
	\label{proof2}
\end{figure}

\vspace{0.1in}
Let us now consider how to adapt the proof to the unrooted Cayley Tree.  
To show that for $t_c < t < 1$ there are no Lee--Yang zeros for the unrooted tree
$\widehat{\Gamma}_n$ at any $\phi_0 \in [0,\phi_e(t))$, we must show that
$\left(\btilde_{\phi_0,t,k+1} \circ \btilde_{\phi_0,t,k}^{n-1}\right)\left(\phi_0\right)<~\pi$; see Proposition \ref{PROP:RENORM} (and also Remark \ref{rmk:restate}). 

Note that even though we will be discussing $\btilde$ indexed with both branching number $k$ and $k+1$, the functions $\theta_\bullet(t)$
and $\phi_e(t)$ used below will correspond only to Formula (\ref{eq:kappa}) with  branching number~$k$.
From the proof for the unrooted tree, we have for any $n$ that  $\btilde^{n-1}_{\phi_0,t,k}(\phi_0) < \theta_\bullet(t)$.
Lemma~\ref{lm:mono} gives that
\begin{align*}
\left(\btilde_{\phi_0,t,k+1} \circ \btilde_{\phi_0,t,k}^{n-1}\right)\left(\phi_0\right) \leq \btilde_{\phi_0,t,k+1}(\theta_\bullet(t)) \leq \btilde_{\phi_e(t),t,k+1}(\theta_\bullet(t)),
\end{align*}
where we have used that 
$0 \leq \phi_0 < \phi_e(t)$ in the second inequality.  Therefore, it suffices to show
that for all $t \in (t_c,1)$ we have that $\btilde_{\phi_e(t),t,k+1}(\theta_\bullet(t)) < \pi$.
An explicit calculation yields that
\begin{align*}
\btilde_{\phi_e(t),t,k+1}(\theta_\bullet(t)) 
= 2\pi-&2\arctan \left( \sqrt {{\frac {-{k}^{2}{t}^{2}-2k{t}^{2}+{k
}^{2}-{t}^{2}-2k+1}{ \left( k+1 \right) ^{2} \left( {t}^{2}-1
 \right) }}} \right) \\  &-2\arccos \left({\frac { \left( k+1 \right) {t}^{2}-k+1}{2t}} \right)
\end{align*}
and also 
\begin{align*}
\frac{d}{dt}  \left(\btilde_{\phi_e(t),t,k+1} (\theta_\bullet(t))\right) = 
2\,{\frac {k-1}{t\sqrt {- \left( k+1 \right) ^{2}{t}^{4}+ \left( 2\,{k}^{2}+2 \right) {t}^{2}- \left( k-1 \right) ^{2}}}},
\end{align*}
which is greater than $0$ for $t \in (t_c,1)$.  Meanwhile, one can directly check that $\btilde_{\phi_e(0),0,k+1}(\theta_\bullet(0)) = 0$ and that $\btilde_{\phi_e(1),1,k+1}(\theta_\bullet(1)) = \pi$.
The result follows.

\noindent
\qed (Theorem A, part (iii).)


\section{Basic Real Dynamics for $B_{z,t}$}
\label{SEC:BASIC_REAL}

Throughout this section we recall some basic facts about expanding circle maps,
i.e.  $f:S^1 \rightarrow S^1$ with constants $C>0, \lambda >1$ so that for all
$x \in S^1$ and for all integers $n\geq 1$, we have $$|(f^n)'(x)|\geq
C\lambda^n.$$ We will call $\lambda$ the \textit{expansion constant} of $f$.
Every expanding circle map is a covering map of degree $d$ with $|d|>1$. The
term absolute continuity will refer to absolute continuity with respect to the
Lebesgue measure on $S^1$.  Most of the theorems stated in this section were
originally proved in more general settings, but we state them in the more
narrow context that will be used here. We also provide references that are
suitable in our context, rather than the original ones.\par

\subsection{Measure of Maximal Entropy (MME) and Absolutely Continuous Invariant Measure (ACIM)}

Fix a continuous function $\phi: S^1 \rightarrow \mathbb{R}$. Let $\mathcal{M}_f$ be the space of all $f$-invariant Borel probability measures on $S^1$. The \textit{pressure} of the system with respect to $\phi$ is
$$P (\phi):=\sup_{\MEAS \in \mathcal{M}_f}\left(h_\MEAS(f)+\int_{S^1} \phi d\MEAS\right),$$
where $h_\MEAS (f)$ denotes the metric entropy of $\MEAS$ under $f$. The measure achieving this supremum (if it exists) is called an \textit{equilibrium state}, which a priori need not be unique. However, in the case when $f$ is expanding and $\phi$ is H\"{o}lder, the equilibrium state is unique:

\begin{uniquenessequilib}[\bf Ruelle \cite{RT}]
	Let $f$ be a $C^2$ expanding circle map and $\phi :S^1 \rightarrow \mathbb{R}$ be a H\"{o}lder function. Then there exists a unique equilibrium state for $\phi$.
\end{uniquenessequilib}

Moreover, we have the following inequality:

\begin{Ruelleineq}{\bf \cite{RI}}
	Let $f$ be a $C^2$ expanding circle map. Suppose $\MEAS \in \mathcal{M}_f$ is ergodic, then
	$$h_{\MEAS}(f) \leq \chi_{\MEAS}(f),$$
	where $\chi_{\MEAS}(f)$ is the Lyapunov exponent of $(f,\MEAS)$, i.e.
	$$\chi_{\MEAS}(f):=\int_{S^1} \log |f'(x)| d\MEAS.$$
\end{Ruelleineq}

As an immediate consequence, we have

\begin{corollary}\label{COR:PRESSURE_LEQ0}
	If $\phi(x)=-\log |f'(x)|$, then $P(\phi)\leq 0$.
\end{corollary}

Next, recall that the \textit{Hausdorff dimension} of a Borel probability measure $\MEAS$ on a compact manifold $M$ is defined as
$$\textrm{HD}(\MEAS):=\inf_{\substack{Y \subset M\\ \MEAS(Y)=1}}\textrm{HD}(Y).$$

\begin{Youngfmla}[\bf \cite{Y}, Section 4.4]
	Let $f$ be a $C^2$ expanding circle map, and let $\MEAS$ be a $f$-invariant ergodic probability measure. Then
	$$h_{\MEAS}(f)=\chi_{\MEAS}(f)\cdot {\rm HD}(\MEAS).$$
\end{Youngfmla}

Set $\phi=0$ in Ruelle's Theorem. Then the unique measure $\eta$ achieving the supremum $\sup_{\MEAS \in \mathcal{M}_f} h_{\MEAS}(f)$ is called the {\it measure of maximal entropy} (MME). The measure $\eta$ can be constructed using a pullback argument: For any $y\in S^1$,
$$\frac{1}{d^n} \sum_{y=f^n (x)}\delta_x \rightarrow \eta.$$
The measure $\eta$ satisfies $f^* \eta = d\cdot \eta$. By the Variational Principle, we have $h_{\eta}(f)=\log d$, the topological entropy of $f$.\par

Meanwhile, $f$ has an {\em absolutely continuous invariant measure} (ACIM) $\nu$, which actually has $C^1$ density with 
respect to Lebesgue measure, see \cite{Sacksteder, Krzy}.

\begin{proposition} Let $f$ be a $C^2$ expanding circle map. If $\MEAS \in \mathcal{M}_f$ is not absolutely continuous, then  
${\rm HD}(\MEAS)<1$. \end{proposition}\label{PROP:HD}

\begin{proof}
The ACIM $\nu$ satisfies ${\rm HD}(\nu)=1$. Therefore the Ledrappier--Young
Formula gives that $h_\nu (f)= \chi_\nu (f).$ On the other hand, Ruelle's
Inequality gives that for any ergodic $\MEAS \in \mathcal{M}_f$:
$$h_{\MEAS}(f)\leq \chi_\MEAS (f),$$ where equality is attained if and only if
$\MEAS$ is the equilibrium state for $\phi=-\log |f'(x)|$, by Corollary~\ref{COR:PRESSURE_LEQ0}. Therefore, the absolutely continuous $\nu$ is the
unique equilibrium state for $\phi$. Since $\MEAS$ is not absolutely continuous,  $\MEAS \neq \nu$, and hence $\MEAS$ must
satisfy $h_{\MEAS}(f)< \chi_\MEAS (f)$. Applying the Ledrappier--Young Formula
again, we have $\textrm{HD}(\rho)~<~1$.
\end{proof}

We will need a special version of Birkhoff's Ergodic Theorem for the ACIM
$\nu$.  For any continuous function $\phi \in C(S^1)$, let
$\phi_n:=\frac{1}{n}\sum_{k=0}^{n-1}\phi \circ f^k$, and
$\overline{\phi}:=\int_{S^1} \phi \ d\nu.$  Since, $\nu$ is ergodic, for $\nu$
a.e.\ $x$ (and hence Lebesgue a.e.\ $x$), we have $\lim_{n\rightarrow \infty}
\phi_n (x)=\overline{\phi}$.  In the proof of Theorem D we will need to give up some control on the sequence $\{\phi_n (x)\}$ in order to have more control on the exceptional set.  

\begin{specialergodictheorem}[Kleptsyn--Ryzhov--Minkov \cite{Kleptsyn}]
	Let $f$ be an expanding circle map and let $\phi \in C(S^1)$.  For any $\epsilon > 0$, the set
	$$K_{\phi,\epsilon} :=\{x \in S^1 : \limsup_{n\rightarrow \infty} |\phi_n (x)-\overline{\phi}|>\epsilon\}$$
	has Hausdorff dimension less than one.
\end{specialergodictheorem}

\begin{proof}
	We verify that $(f,\nu)$ satisfies the hypotheses from Theorem 1 of \cite{Kleptsyn}.
Since $\nu$ is absolutely continuous, it is the global SBR measure.
	Meanwhile, the work of Kifer \cite{Kifer} and Young \cite{Y1} implies that $\nu$ satisfies a Large Deviations Principle.
	For example, while Theorem 3 from \cite{Y1} is stated for the equilibrium state of an Axiom A attractors of a $C^2$ diffeomorphism, one can check that each 
	line of the proof adapts directly to expanding maps of the circle.
\end{proof}

\subsection{Conjugacies Between Expanding Circle Maps}

\begin{Shubthm}{\bf \cite{S}}
	For all $\epsilon >0$, there exists $\delta>0$ so that if $f, g$ are $C^1$ expanding circle maps with $||f-g||_{C^1} < \delta$, then $f$ is $C^0$ conjugate to $g$ via some $h$ with $||h-\textrm{id}||_{C^0}<\epsilon.$
\end{Shubthm}

\begin{SSthm}{\bf \cite{PS}}
	Two $C^r, r\geq 2$, expanding maps of the circle which are absolutely continuously conjugate are $C^r$ conjugate. If two Blaschke products are absolutely continuously conjugate on the unit circle $\mathbb{T}$, then they are conjugate by a M\"{o}bius transformation of the Riemann Sphere.
\end{SSthm}

By Shub's Theorem, any two $C^2$ expanding maps of the circle are topologically conjugate. If they are also $C^2$ close, then we have more control over the H\"{o}lder exponent of the conjugacy.

\begin{proposition}\label{PROP:HOLDERCONJ}
Let $f$ be a $C^2$ expanding circle map.  Then, for any $\epsilon >0$, there exists $\delta >0$ such that
if $g$ is another $C^2$ expanding circle map with $||f-g||_{C^2}<\delta$, then $f$ and $g$ are H\"{o}lder conjugate by $h$ with exponent  $1-\epsilon$ and multiplicative constant independent of $g$.
\end{proposition}

The proof of Proposition \ref{PROP:HOLDERCONJ} requires the following Denjoy
style distortion estimate. A proof can be found in \cite{PS}.  Given a covering
map $f:S^1 \rightarrow S^1$, let $\tilde{f}:\mathbb{R}\rightarrow \mathbb{R}$
be any lift of $f$. 
\begin{lemma}\label{DENJOY}
	Let $f$ be a $C^2$ expanding circle map. For all $L>0$, there exists $M>0$, so that the following is true: for any interval $I\subset S^1$, let $N$ be the first iterate so that $|\tilde{f}^N(I)|>L$, then ${\rm dist}_I(f^N)<M$.
\end{lemma}
\noindent
Here, $\textrm{dist}_I(f)$ is the \textit{distortion} of $f$ on $I$ defined by
$$\textrm{dist}_I(f):=\sup_{x,y \in I}\log \frac{|f'(x)|}{|f'(y)|}.$$

Let $p$ be a periodic point of $f$ with period $n$. Define $\lambda_{avg}(p)=((f^n)'(p))^{1/n}.$ Remark that for all periodic points $p$ of an expanding circle map $f$ with expansion constant $\lambda$, we have $\lambda_{avg}(p) \geq \lambda$.

\begin{proof}[Proof of Proposition \ref{PROP:HOLDERCONJ}]
	For any $\alpha >1$, it is straightforward to see that there exists $\delta >0$ such that whenever $||f-g||_{C^2} < \delta$, we have for all periodic points $p$ of $f$ 
	$$\frac{1}{\alpha} \lambda_{avg}(p) \leq \lambda_{avg} (h(p)) \leq \alpha \lambda_{avg}(p),$$
	where $h$ is the $C^0$ conjugacy close to the identity given by Shub's Theorem.	We will show that $h$ is in fact H\"{o}lder continuous with exponent $\left(1-\frac{\log \alpha}{\log \lambda}\right)$, where $\lambda$ is the expansion constant of~$f$.

	Let $I \subset S^1$ be an interval. Then there exists an integer $N \geq 0$ so that $f^N$ is the first iterate so that $|\tilde{f}^N(I)|>4\pi$, where $\tilde{f}$ is any choice of lift of $f$. Then, since $h$ is a conjugacy between $f$ and $g$, we have that $g^N$ is also the first iterate so that $\tilde{g}^N(\tilde{h}(I))>4\pi$, with $\tilde{h}$ and $\tilde{g}$ suitable lifts of $h$ and~$g$. By the Intermediate Value Theorem, $f^N$ has a fixed point $p \in I$, corresponding to a periodic point of $f$ with period dividing $N$. Then $h(p) \in h(I)$ is a periodic point of $g$ of the same period. Then, using the Distortion Lemma, we have
\begin{align*}
|f^N(I)| &\asymp \lambda_{avg}(p)^N|I|,\\
|g^N(h(I))| &\asymp \lambda_{avg}(h(p))^N |h(I)|.
\end{align*}
(The asymptotic notation $\asymp$ means the ratio of the left and right sides is bounded from above and below by positive
constants, independent of $N$ (or, equivalently, independent of $|I|$, since $N$ depends on $|I|$).
Since $4\pi \leq |f^N(I)| \leq 4\pi K$ and $4\pi \leq |g^N(h(I))| \leq 4\pi K$, where
$$K=\max \left\{\max_{x\in S^1} f'(x),\max_{x\in S^1} g'(x) \right\},$$ we have $|f^N(I)|\asymp |g^N(h(I))| \asymp 1$ and thus
	\begin{equation*} \lambda_{avg}(p)^N |I| \asymp \lambda_{avg}(h(p))^N |h(I)| \asymp 1.
	\end{equation*}
	Then $$|h(I)| \asymp \frac{1}{\lambda_{avg}(h(p))^N}=\frac{1}{(\lambda_{avg}(p)^N)^\gamma} \asymp |I|^\gamma,$$
	where $\gamma=\frac{\log \lambda_{avg}(h(p))}{\log \lambda_{avg}(p)}>\frac{\log (\lambda_{avg}(p)/\alpha)}{\log \lambda_{avg}(p)} \geq 1-\frac{\log \alpha}{\log \lambda}.$
	Therefore, $|h(I)|\asymp |I|^\gamma < |I|^{1-\frac{\log \alpha}{\log \lambda}}.$ This shows that $h$ is H\"{o}lder with exponent $\left(1-\frac{\log \alpha}{\log \lambda}\right)$. Lastly, the reader can check that all the multiplicative constants that are implicit in the $\asymp$ notation (including those coming from the distortion lemma) can be made uniform in $g$.
\end{proof}




\subsection{Applications to the Blaschke Products $\bzt$}
Recall that for $0<t<1$  and $z\in {\rm interior}(S_t)$, the map $B_{z,t}(w)=z\left(\frac{w+t}{1+wt}\right)^k$ is expanding on the unit circle $\mathbb{T}$. On the other hand, the critical points of $B_{z,t}$, $w=-t$ and $w=-\frac{1}{t}$, are not fixed, and hence lie in the basins of attraction of some other fixed points in $\mathbb{D}$ and $\hat{\mathbb{C}}\setminus \overline{\mathbb{D}}$.

\begin{lemma}\label{MMELESSTHAN1}
	Fix branching number $k \geq 2$, fix $0<t<1$, and $z \in {\rm interior}(S_t).$ Then the MME $\eta_{k,z,t}$ for $B_{k,z,t}$ satisfies $\textrm{HD}(\eta_{k,z,t})<1$.
\end{lemma}

\begin{proof}
To simplify notation, we write $B \equiv B_{k,z,t}$ and $\eta \equiv \eta_{k,z,t}$.
Assume for contradiction that $\eta$ is absolutely continuous, so $\eta= m(\theta)
d\theta$, where $m(\theta)$ is $L^1$ and $d \theta$ denotes the Lebesgue measure
on $\mathbb{T}$.  By absolute continuity, $\textrm{HD}(\eta)=1$, so the
Ledrappier--Young Formula gives that $h_{\eta}(B)=\chi_{\eta}(B)$. Then, by
Ruelle's Inequality and Ruelle's Theorem, $\eta$ must be the unique measure
achieving the equilibrium state for the potential $-\log B'$. It follows from \cite{Sacksteder,Krzy} that
the density function $m(\theta)$ of $\eta$ is $C^1$. Define $h:\mathbb{T} \rightarrow
\mathbb{T}$ by
$$h(\theta)=\int_0^\theta m(\psi)d\psi,$$
we will show that $m(\theta)\neq 0$ for all $\theta$. If $m(\theta_0)=0$ for some $\theta_0$, then invariance of $m$ implies that $m$ vanishes on the backward orbits of $\theta_0$, which is dense in $\mathbb{T}$ because $B$ is expanding, therefore $m(\theta) \equiv 0$, contradiction. This shows that $m(\theta) \neq 0$ for all $\theta$. Hence $h$ is in fact a $C^2$ diffeomorphism. On the other hand, the MME $\eta$ satisfies $B^*\eta=k\cdot \eta$, then we have
\begin{align*}
h(B(\theta))&=\eta((0,B(\theta))
                  =k\cdot \eta(0,\theta)
                  =k\cdot h(\theta).
                \end{align*}
Therefore, $B$ and $\theta \mapsto k \cdot \theta$ are conjugate by the $C^2$ mapping $h$.

Then, by the Shub--Sullivan Theorem, $h$ can be extended to a M\"{o}bius transformation of the
Riemann Sphere $\hat{\mathbb{C}}$ so that it conjugates $B_{k,z,t}$ to $z \mapsto z^k$ on $\hat{\mathbb{C}}$.  However, both of the critical points of $z\mapsto z^k$ are fixed points, but, for $0 < t < 1$, the critical point 
$t$ for $B_{k,z,t}$ is not a fixed point. Therefore
$\eta_{z,t}$ cannot be absolutely continuous, then by Proposition
\ref{PROP:HD}, $\textrm{HD}(\eta_{z,t})<1$.
	
\end{proof}

\begin{remark}
Some of the key statements above can also be obtained from the
perspective of complex dynamics.  In particular, Lemma \ref{MMELESSTHAN1}
is a direct application of Zdunik's Theorem \cite{Z} and, in the case of
Blaschke Products $B_{z,t}$, Proposition \ref{PROP:HOLDERCONJ} about
H\"older regularity of the conjugacy follows from theory of holomorphic motions
\cite{BR}.  However, we found that presenting them using real dynamics was more concrete for the purposes here.
\end{remark}

\begin{proposition}\label{EXPLICITLYAP}
	For $z\in {\rm interior}(S_t)$, the Lyapunov exponent of $B_{z,t,k}$ with respect to the ACIM $\nu_{z,t}$ is given by
\begin{align*}
	\chi_{\nu_{z,t}}(B_{z,t,k}) = 2\pi\log \left|\frac{k(1-t^2)w_{\mathbb{D}}(1-{w_{\mathbb{D}}}t)}
         {(w_{\mathbb{D}}+t) (1+w_{\mathbb{D}}t)(t-w_{\mathbb{D}}) }\right|,
\end{align*}
	where $w_{\mathbb{D}}$ is the unique attracting fixed point of $B_{z,t}$ in $\mathbb{D}$.

\end{proposition}

\begin{proof}
	Let $\psi$ be the disc automorphism $$\psi
	(w)=\frac{w-w_{\mathbb{D}}}{1-\overline{w_\mathbb{D}}w}.$$ Then the map
	$\check{B}_{z,t}:=\psi \circ B_{z,t} \circ \psi^{-1}$ is a Blaschke product
	expanding on $\mathbb{T}$ with the attracting fixed point at the origin, hence
	its ACIM $\check{\nu}_{z,t}$ is the normalized Lebesgue measure $d\theta$ \cite{PUJALS}. The only
	critical point of $B_{z,t}$ in the unit disc is $w=t$, which is mapped to
	$\frac{t-w_{\mathbb{D}}}{1-\overline{w_{\mathbb{D}}t}}$ under $\psi$. By
	Jensen's Formula, we have
	\begin{align*}
	\int_{\mathbb{T}} \log |B'_{z,t}|d\nu_{z,t}= \int_{\mathbb{T}} \log |\check{B}_{z,t}'|d\theta&=2\pi \left(\log |\check{B}_{z,t}'(0)|-\log \left|\frac{t-w_{\mathbb{D}}}{1-\overline{w_{\mathbb{D}}t}}\right|\right)\\
	&=2\pi\log \left|\frac{k(1-t^2)(w_{\mathbb{D}}+t)^{k-1} (1-\overline{w_{\mathbb{D}}}t)}{(1+w_{\mathbb{D}}t)^{k+1}(t-w_{\mathbb{D}}) }\right|.
	\end{align*}
Using that $B_{z,t}(w_{\mathbb{D}}) = w_{\mathbb{D}}$, that $t$ is real, and
that $|z| = 1$, this formula simplifies to the stated one. 
\end{proof}



\section{Partial Hyperbolicity and Transverse Invariant Measures}
\label{SEC:PH}
Fix a temperature $t \in (0,1)$.  This section is devoted to proving several properties of the skew product mapping
\begin{align*}
B(\phi,\theta):=(\phi, B_{\phi,t}(\theta)).
\end{align*}
To simplify notation, let $\mathbb{T}_\phi:=\{\phi\}\times \mathbb{T}$, and $x:=(\phi,\theta)$ when there is no ambiguity.

\subsection{Partial Hyperbolicity}
Let $X_t \subseteq {\rm interior}(S_t)$ be any compact interval. One says that $B~:~X_t \times \mathbb{T} \rightarrow X_t \times \mathbb{T}$ is {\em partially hyperbolic} if
\begin{enumerate}
\item $B$ has a ``vertical'' tangent conefield $\mathcal{K}^v(x)$ that depends continuously on $x$ and is invariant under $DB$, 
\item $B$ has a ``horizontal'' linefield $\mathcal{L}^c(x)$ 
lying in the horizontal cone $\mathcal{K}^h(x)$ that is complementary to $\mathcal{K}^v(x)$,
that depends continuously on $x$, and that is invariant under $DB$.

\item Vertical tangent vectors $v \in \mathcal{K}^v (x)$ get exponentially stretched under $DB^n$
at a rate that dominates the growth rate of any horizontal vector $w \in \mathcal{L}^c(x)$.

\end{enumerate}
The line field $\mathcal{L}^c(x)$ satisfying the properties stated in (2) is called the {\em central linefield} of $B$.

Since $B$ is a skew product over the identity, with the restriction to each
vertical fiber being an expanding map of the circle, it is a straightforward
application of techniques from smooth dynamics to prove that $B$ is 
partially hyperbolic on $X_t \times \mathbb{T}$.  Moreover, the cones $\mathcal{K}^v (x)$ are actually vertical, i.e.\ they contain the tangent lines to the vertical fibers $\mathbb{T}$.
We remark that the central linefield  is given by
$$\mathcal{L}^c(x)=\bigcap_{n\geq 0} (DB^n)^{-1} \mathcal{K}^h (B^n(x)).$$
A smooth
curve is called {\it central} if it is tangent to the central linefield
$\mathcal{L}^c$. A {\it central foliation} $\mathcal{F}^c$ is a strictly
horizontal foliation of $X_t \times \mathbb{T}$ invariant under $B^{-1}$.

\begin{proposition}\label{PROP:CENTRAL_FOL_UNIQUE}
	There exists a unique central foliation $\mathcal{F}^c$.
\end{proposition}

\begin{proof}
	It follows from the Peano Existence Theorem that continuous linefields are integrable, so we can find a central curve through any point $x \in X_t \times \mathbb{T}$ by integrating the central linefield $\mathcal{L}^c$. Denote $\mathcal{F}^c$ the collection of all such curves. Let us show that $\mathcal{F}^c$ is a central foliation. \par 

	Suppose $\gamma_1\neq\gamma_2$ are two central curves which overlap.
Let us parameterize $\gamma_1,\gamma_2$ by $\phi$, and let
$\gamma_1(\phi_0)=\gamma_2(\phi_0)$ be a boundary point of the intersection.
Pick $\phi_1\neq \phi_0$ arbitrarily close to $\phi_0$ so that
$\gamma_1(\phi_1)\neq \gamma_2(\phi_1)$, and consider the genuine vertical
segment $I$ between $\gamma_1(\phi_1)$ and $\gamma_2(\phi_1)$. Since $B_\phi$
is expanding for every $\phi$, the length of $I$
is exponentially stretched under the iterates of~$B$. Since $I$ can be chosen arbitrarily close to
$\gamma_1(\phi_0)$, this implies that at the point $B^n(\gamma_1(\phi_0))$, at
least one of the curves, $B^n(\gamma_1)$ or $B^n(\gamma_2)$, lies in the
vertical cone when $n$ is sufficiently large, which contradicts the invariance
of $\mathcal{L}^c$.

	This proves that the linefield $\mathcal{L}^c$ is uniquely integrable, so the family $\mathcal{F}^c$ of all integral curves forms the central foliation. 
	
\end{proof}

\subsection{Transverse invariant measure}
Let $\tau_1$ and $\tau_2$ be two local transversals to $\mathcal{F}^c$.  We say
that $\tau_1$ and $\tau_2$ {\em correspond under $\mathcal{F}^c$-holonomy} if
for any $x_1 \in \tau_1$, there exists a unique $x_2 \in \tau_2$ so that $x_1$
and $x_2$ belong to the same leaf of $\mathcal{F}^c$, and vice versa. In this
case, the {\it holonomy transformation} $g_{\tau_1,\tau_2}:\tau_1 \rightarrow
\tau_2$ is defined by the property that $(\phi,\theta)\in \tau_1$ and
$g_{\tau_1,\tau_2} (\phi,\theta) \in \tau_2$ belong to the same leaf. In the
case that one or both of the transversals is a vertical fiber
$\mathbb{T}_{\phi}$, we will abuse notation and write
$g_{\phi_1,\phi_2}:=g_{\mathbb{T}_{\phi_1},\mathbb{T}_{\phi_2}}$.\par

A {\it transverse invariant measure} $\eta$ for $\mathcal{F}^c$ is a family of measures $\{\eta_\tau: \tau \ \textrm{tranversal to} \ \mathcal{F}^c\}$, such that for any $\tau_1, \tau_2$ which correspond under $\mathcal{F}^c$-holonomy, we have $\eta_{\tau_2} = (g_{\tau_1,\tau_2})_* (\eta_\tau)$. Such measures can be specified by a single measure $\eta_\tau$ on a global transversal $\tau$. \par

The main two results of this section are:
\begin{proposition}\label{TRANSINVARMAINPROP}
	There is a transverse invariant measure $\eta$ for $\mathcal{F}^c$, so that
\begin{itemize}
\item[{\rm (i)}] on each vertical fiber $\mathbb{T}_\phi$, $\eta$ is the MME $\eta_{\phi,t}$ for $B_{\phi,t}$, and
\item[{\rm (ii)}] on the diagonal $$\Delta_t:=\{(\delta,\delta): \delta \in X_t \}\subset X_t \times \mathbb{T},$$ $\eta$ is the limiting measure of the Lee--Yang zeros $\mu_t$ for the rooted Cayley Tree.
\end{itemize}
\end{proposition}

\begin{proposition}\label{PROP:ROOTED_AND_FULL_SAME_MEASURE}
The full and rooted Cayley Tree have the same limiting measure of Lee--Yang zeros.
I.e.\ for any $t \in [0,1]$ we have $\widehat{\mu}_t = \mu_t$.
\end{proposition}

The proof of both propositions will follow from the next three lemmas. \par
The Lee--Yang zeros for the $n$-th level rooted Cayley Tree are obtained by pulling back the horizontal line $\{\theta=\pi\}$ by $B^n$ and then intersecting with the diagonal $\Delta_t$. Therefore the following lemma is of particular interest.

\begin{lemma}\label{LEM:DIAGONAL}
	The leaves of the central foliation $\mathcal{F}^c$ have negative slope.  In particular,
	the diagonal $\Delta_t$ is transversal to $\mathcal{F}^c$.
\end{lemma}

\begin{proof}
	One can prove by induction that $$DB^n(\phi,\theta)=\begin{bmatrix} 
	\ \ 1     & 0\\ 
	\ \ * & (B_\phi^n)'(\theta)
	\end{bmatrix},$$
	where $*$ denotes a positive term. So we see that the vector $v=[1,0]^T$ is vertically stretched exponentially by $DB^n$, while its horizontal length is fixed. Then, for some integer $n>0$, $DB^nv$ is in the vertical cone. This implies $v \notin \mathcal{L}^c$. So the leaves of $\mathcal{F}^c$ must have negative slopes. Hence $\Delta_t$ is transversal to $\mathcal{F}^c$.
\end{proof}

Denote the cardinality of a set $S$ by $\#S$.

\begin{lemma}\label{LEM:TRANSVERSEMEASURE}
	Let $I_\phi\subset \mathbb{T}_\phi$ be an arbitrary interval and let $\tau$ be a local transversal to $\mathcal{F}^c$
so that $I_\phi$ and $\tau$ correspond under $\mathcal{F}^c$-holonomy.  For any smooth curve $\gamma$ that can be represented as a graph of a smooth function $\gamma(\phi):X_t\rightarrow \mathbb{T}$, we have
	$$\lim_{n\rightarrow \infty} \frac{1}{k^n} \#\left\{B^{-n}(\gamma)\cap \tau\right\}=\eta_\phi (I_\phi), $$
	where $\eta_\phi$ is the MME for $B_{\phi,t}$.
\end{lemma}

\begin{proof}
Let $\tilde{B}:X_t\times \mathbb{R} \rightarrow X_t \times \mathbb{R}$
be a lift of $B$ and let $\widetilde{\gamma}_0$ be a lift of
$\gamma$ to $X_t \times \mathbb{R}$.  Meanwhile, let $\widetilde{\gamma}$ denote the union of all lifts of $\gamma$,
i.e.\ $\widetilde{\gamma}$ is the union of vertical translations of $\widetilde{\gamma}_0$ by every integer multiple of $2\pi$.

Let $\pi_2: X_t \times \mathbb{T} \rightarrow \mathbb{T}$ and $\tilde{\pi}_2: X_t \times \mathbb{R} \rightarrow \mathbb{R}$
denote the projections onto the second coordinate.

Suppose $\xi \subset X_t \times \mathbb{R}$ is any smooth curve.  Let 
$\ell(\xi):=|\tilde{\pi}_2 (\xi)|$ denote the Euclidean
length of its projection to the second coordinate $\mathbb{R}$.
For any integer $n>0$, the quantity $\ell(\tilde{B}^n(\tau))/2\pi$
is the number of times the curve $\pi_2(B^n(\tau))$ wraps around the circle $\mathbb{T}$, which
differs from the number of times $\tilde{B}^n(\tau)$ intersects
$\widetilde{\gamma}$ by at most $1+\frac{M}{2\pi}$, where $M=\ell(\widetilde{\gamma}_0)$, i.e.
	$$\left|\frac{\ell(\tilde{B}^n(\tau))}{2\pi}-\#\left\{\tilde{B}^n(\tau)\cap \widetilde{\gamma}\right\} \right|\leq 1+\frac{M}{2\pi}.$$
	On the other hand, since $B$ maps each leaf of $\mathcal{F}^c$ to another leaf of $\mathcal{F}^c$ , there exists a constant $A>0$ such that for all $n>0$,
	$$\left|\ell(\tilde{B}^n(I_\phi))-\ell(\tilde{B}^n(\tau))\right|<A.$$ 
	Thus, 
	$k^{-n} \left|\ell(\tilde{B}^n(I_\phi))-\ell(\tilde{B}^n(\tau))\right|\rightarrow 0$ as $n\rightarrow \infty$.  Therefore, 
we have
	$$\lim_{n\rightarrow \infty} \frac{1}{k^n} \left|\#\left\{ B^{-n}(\gamma)\cap I_\phi\right\}-\#\left\{ B^{-n}(\gamma)\cap \tau\right\}\right|=0.$$
	Finally, notice that $\lim_{n \rightarrow \infty}k^{-n} \#\{B^{-n}(\gamma)\cap I_\phi\}$ converges
to the measure of $I_\phi$ under the measure of maximal entropy $\eta_\phi$ for $B_{\phi,t}$.
In particular, 
	$$\lim_{n\rightarrow \infty}\frac{1}{k^n}\#\left\{ B^{-n}(\gamma)\cap I_\phi\right\}=\lim_{n\rightarrow \infty} \frac{1}{k^n} \#\left\{B^{-n}(\gamma)\cap \tau\right\}=\eta_\phi(I_\phi).$$
\end{proof}

Suppose $\tau_1$ and $\tau_2$ correspond under $\mathcal{F}^c$-holonomy.  Then, for any $\phi \in X_t$ 
we can find $I_\phi\subset \mathbb{T}_\phi$ so that $I_\phi$ corresponds to both $\tau_1$ and $\tau_2$ under 
$\mathcal{F}^c$-holonomy.  Therefore, the following is a direct consequence of Lemma \ref{LEM:TRANSVERSEMEASURE}.
\begin{lemma}\label{PROP:TRANSVERSEMEASURE}
	For any two transversals $\tau_1$ and $\tau_2$ that correspond under $\mathcal{F}^c$-holonomy
we have 
	$$\lim_{n\rightarrow \infty}\frac{1}{k^n}\#\left\{ B^{-n}(\gamma)\cap \tau_1\right\}=\lim_{n\rightarrow \infty} \frac{1}{k^n} \#\left\{B^{-n}(\gamma)\cap \tau_2\right\}.$$
\end{lemma}

\vspace{0.1in}

\begin{proof}[Proof of Proposition \ref{TRANSINVARMAINPROP}]
For each $\tau$ transversal to $\mathcal{F}^c$, Lemma
\ref{LEM:TRANSVERSEMEASURE} defines a premeasure on the algebra generated by
intervals on $\tau$. By the Carath\'eodory Extension Theorem, this premeasure
can be extended to a Borel measure. Lemma \ref{PROP:TRANSVERSEMEASURE} gives
that the measure is invariant under the holonomy transformation. It follows
from the construction above that $\eta$ is the MME $\eta_{\phi,t}$ on each
vertical fiber $\mathbb{T}_\phi$. By taking $\gamma(\phi)\equiv\pi$ and
$\tau=\Delta_t$, it follows that on the diagonal $\Delta_t$, $\eta$ is the
limiting measure of the Lee--Yang zeros for the rooted Cayley Tree. \par
\end{proof}

\begin{proof}[Proof of Proposition \ref{PROP:ROOTED_AND_FULL_SAME_MEASURE}:]
Fix $0 \leq t < 1$ in order to show that $\widehat{\mu_t} = \mu_t$, where
$\widehat{\mu}_t$ denotes the limiting measure of Lee--Yang zeros for the
full Cayley Tree and temperature $t$.  Let us now emphasize the branching
number in the notation for the skew product by writing $B_{k}(\phi,\theta) =
(\phi, B_{\phi,t,k}(\theta))$.  Then $B_{k+1}^{-1}(\theta=\pi)$ is a disjoint
union $\bigcup_{i=1}^{k+1} \gamma_i$ of $k+1$ smooth curves, each can be
represented as a graph $\gamma_i:X_t \rightarrow \mathbb{T}$. For any $\phi_1 < \phi_2$ let
$\Delta_{\phi_1,\phi_2}$ be the part of the diagonal $\Delta$ whose first
coordinate satisfies $\phi_1 \leq \phi \leq \phi_2$.  Then, Lemma
\ref{LEM:TRANSVERSEMEASURE} gives that
	$$\lim_{n\rightarrow \infty} \frac{1}{k^n(k+1)} \#\left\{B_k^{-n}\circ B_{k+1}^{-1}(\theta=\pi)\cap \Delta_{\phi_1,\phi_2}\right\}=\eta_\phi (I_\phi) = \lim_{n\rightarrow \infty} \frac{1}{k^n} \#\left\{B_k^{-n}(\theta=\pi)\cap \Delta_{\phi_1,\phi_2}\right\},$$
where $I_\phi$ is the holonomy image of $\Delta_\delta$ along $\mathcal{F}^c$.
Meanwhile, it follows from Lemma \ref{PROP:RENORM} that the left hand side of the equation is 
$\widehat{\mu_t}([\phi_1,\phi_2])$ and the right hand side is $\mu_t([\phi_1,\phi_2])$.
\end{proof}

As we mentioned in the introduction, from now on we will ignore the distinction between the full Cayley Tree and the rooted Cayley Tree, so the term ``Cayley Tree'' will refer to either of them.\par

\subsection{Regularity of the holonomy along $\mathcal{F}^c$.}
We finish this section with the next two results which will be used in the proof of Theorems B and C.


\begin{proposition}\label{LEM:HOLHOLDER}
Let $\phi_1 \in X_t$. Then for any $\epsilon > 0$, there exists $\delta > 0$ so that if
$|\phi_1 - \phi_2|~\leq~\delta$,  then the holonomy transformation $g_{\phi_1,\phi_2}$ is H\"{o}lder continuous with exponent $1-\epsilon$ and multiplicative constant independent of $\phi_2$.
\end{proposition}

\begin{proof}
Since $B$ maps leaves of $\mathcal{F}^c$ to leaves of $\mathcal{F}^c$, the holonomy transformation $g_{\phi_1,\phi_2}$
provides a conjugacy between $B_{\phi_1}$ and $B_{\phi_2}$.  The result then follows from Proposition \ref{PROP:HOLDERCONJ}.
\end{proof}

Parameterize the diagonal $\Delta_t$ by the horizontal coordinate $\phi \mapsto
(\phi,\phi)$.   Let $\phi_1 \in X_t$ and $\delta > 0$ be chosen so that the
sub-interval $\Delta_\delta \equiv \Delta_{\phi_1,\delta} \subset \Delta_t$
consisting of those points whose first coordinate lies in
$[\phi_1,\phi_1+\delta]$ satisfies $\Delta_{\phi_1,\delta} \subset X_t \times
\mathbb{T}$.  Let $I_\delta \equiv I_{\phi_1,\delta} \subset \mathbb{T}_{\phi_1}$
be the vertical interval that corresponds to $\Delta_{\phi_1,\delta}$ under
$\mathcal{F}^c$-holonomy.  See Figure \ref{FIG_DIAGONAL_HOL} for an
illustration.

\begin{figure}
\begin{picture}(0,0)%
\includegraphics{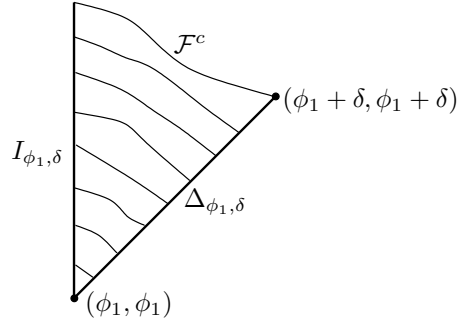}%
\end{picture}%
\setlength{\unitlength}{3947sp}%
\begingroup\makeatletter\ifx\SetFigFont\undefined%
\gdef\SetFigFont#1#2#3#4#5{%
  \reset@font\fontsize{#1}{#2pt}%
  \fontfamily{#3}\fontseries{#4}\fontshape{#5}%
  \selectfont}%
\fi\endgroup%
\begin{picture}(1734,2027)(2979,-4841)
\put(2994,-3822){\makebox(0,0)[lb]{\smash{{\SetFigFont{10}{12.0}{\familydefault}{\mddefault}{\updefault}{\color[rgb]{0,0,0}$I_{\phi_1,\delta}$}%
}}}}
\put(4077,-4125){\makebox(0,0)[lb]{\smash{{\SetFigFont{10}{12.0}{\familydefault}{\mddefault}{\updefault}{\color[rgb]{0,0,0}$\Delta_{\phi_1,\delta}$}%
}}}}
\put(4014,-3151){\makebox(0,0)[lb]{\smash{{\SetFigFont{10}{12.0}{\familydefault}{\mddefault}{\updefault}{\color[rgb]{0,0,0}$\mathcal{F}^c$}%
}}}}
\put(4698,-3493){\makebox(0,0)[lb]{\smash{{\SetFigFont{10}{12.0}{\familydefault}{\mddefault}{\updefault}{\color[rgb]{0,0,0}$(\phi_1+\delta,\phi_1+\delta)$}%
}}}}
\put(3458,-4777){\makebox(0,0)[lb]{\smash{{\SetFigFont{10}{12.0}{\familydefault}{\mddefault}{\updefault}{\color[rgb]{0,0,0}$(\phi_1,\phi_1)$}%
}}}}
\end{picture}%
\caption{Configuration of transversals $I_\delta$ and $\Delta_\delta$ from Lemma \ref{LEM:DIAGONALHOLDER}.\label{FIG_DIAGONAL_HOL}}
\end{figure}

\begin{lemma}\label{LEM:DIAGONALHOLDER}
For any $\phi_1 \in X_t$ and any $0<\gamma<1$, there exists $\delta >0$ such that \\ $g_{I_{\phi_1,\delta}, \Delta_{\phi_1,\delta}}:I_{\phi_1,\delta} \rightarrow \Delta_{\phi_1,\delta}$ is H\"older continuous with exponent $\gamma$.
\end{lemma}

\begin{proof}
Given $\phi_1 \in X_t$ and $0 < \gamma < 1$, let $\delta > 0$ be as in Proposition
\ref{LEM:HOLHOLDER}. 
Let $x, y \in I_{\phi_1,\delta}$ be arbitrary
(assuming $x$ is above $y$ on $I_{\phi_1,\delta}$). Furthermore, let
$(\phi_2,\phi_2):=g_{I_{\phi_1,\delta},\Delta_{\phi_1,\delta}}(x)$.  Since $|\phi_1-\phi_2| \leq
\delta$, Proposition \ref{LEM:HOLHOLDER} gives that for all $x_1, x_2 \in
\mathbb{T}_{\phi_1}$, 
	
	\begin{equation*}
		|g_{\phi_1,\phi_2}(x_1)-g_{\phi_1,\phi_2}(x_2)| < C |x_1-x_2|^{\gamma},
	\end{equation*}
	where $C$ is independent of $\phi_2$. Then, since the leaves of $\mathcal{F}^c$ have negative slopes (see Lemma \ref{LEM:DIAGONAL}), we have
	\begin{equation*}
		|g_{I_\delta,\Delta_\delta}(x)-g_{I_\delta,\Delta_\delta}(y)|< |g_{\phi_1,\phi_2}(x)-g_{\phi_1,\phi_2}(y)|<C |x-y|^{\gamma}.
	\end{equation*}

\end{proof}


\section{Proof of Theorem B}
\label{SEC_PF_THMB}

The following lemma is well-known and its proof is a direct application of the definitions:

\begin{lemma}\label{LEM:HOLDERHAUSDORFF}
Let $f:X\rightarrow Y$ be an $\alpha$-H\"{o}lder continuous function between
metric spaces. Then $\textrm{HD}(f(A))\leq
\frac{1}{\alpha}\textrm{HD}(A)$ for any $A \subset X$.
\end{lemma}


\begin{proof}[Proof of Theorem B]

We must show that the restriction of $\mu_t$ to any compact 
interval $X_t \subset {\rm interior} (S_t)$ has Hausdorff dimension less than one.

Let $Y_t$ be some compact interval such that $X_t \subset {\rm interior} (Y_t)$ and $Y_t \subset {\rm interior} (S_t)$.
We will show that for any $\phi \in Y_t$ there exists $\delta > 0$ so that the restriction of $\mu_t$ to $(\phi,\phi+\delta)$ has
Hausdorff dimension less than one.  Such intervals form an open cover of $X_t$, so the desired result will follow by compactness.
	
Let $\phi \in Y_t$. By Lemma \ref{MMELESSTHAN1}, the MME $\eta_\phi$ of $B_{\phi,t}$ satisfies
$\textrm{HD}(\eta_\phi)=H_\phi <1$, so we can choose $0< \gamma <1$ with
$H_\phi/\gamma <1$. Then, by Lemma \ref{LEM:DIAGONALHOLDER}, there exists
$\delta > 0$ such that $g_{I_{\phi,\delta}, \Delta_{\phi,\delta}}$ is H\"older with exponent
$\gamma$, where $I_{\phi,\delta} \subset \mathbb{T}_\phi$.  (See the paragraph before the statement of Lemma \ref{LEM:DIAGONALHOLDER} for the definitions of $\Delta_{\phi,\delta}$ and $I_{\phi,\delta}$.)
    
Denote $\eta_{I_{\phi,\delta}}$ the restriction of $\eta_\phi$ onto
$I_{\phi,\delta}$. Then $\textrm{HD}(\eta_{I_{\phi,\delta}})\leq H_\phi$, so there is
a subset $E_\delta \subset I_{\phi,\delta}$ with full measure in $I_{\phi,\delta}$ (w.r.t.
$\eta_{I_{\phi,\delta}}$) satisfying $\textrm{HD}(E_\delta)\leq H_\phi$. 
	 
	 Under holonomy the set $g_{I_{\phi,\delta},\Delta_{\phi,\delta}}(E_\delta)$ has 
	 \begin{enumerate}
\item Full measure in $\Delta_{\phi,\delta}$ with respect to the Lee--Yang measure $\eta|_{\Delta_{\phi,\delta}} = \mu_t|_{\Delta_{\phi,\delta}}$,
\item $\textrm{HD}\left(g_{I_{\phi,\delta},\Delta_{\phi,\delta}}(E_\delta)\right) < 1$ (Lemma \ref{LEM:HOLDERHAUSDORFF}).
	 \end{enumerate}

\end{proof}


\section{Proof of Theorem C}
\label{SEC:PF_THMC}
Recall that for each $\phi \in {\rm interior}(S_t)$, $\nu_{\phi,t}$ denotes the
ACIM for the expanding circle map $B_{\phi,t}$. Let
$\chi_{\phi,t}\equiv\chi_{\nu_{\phi,t}}(B_{\phi,t})$ denote the Lyapunov exponent of $\nu_{\phi,t}$.  The following
proposition is a direct application of the Special Ergodic Theorem for the
observable $\log |B_{\phi,t}'|$.

\begin{proposition}\label{PROPNINE}
	For all $\epsilon>0$, $\phi \in {\rm interior}(S_t)$, there exists $G_{\phi}(\epsilon) \subset \mathbb{T}_\phi$ such that
	\begin{enumerate}
		\item $\rm{HD}(\mathbb{T}_\phi \setminus G_{\phi}(\epsilon))<1$.
		\item Each $\theta \in G_{\phi}(\epsilon)$ satisfies $\limsup_{n\rightarrow \infty} \left|\frac{1}{n} \log |(B^n_{\phi,t})'(\theta)| -\chi_{\phi,t}\right| <\epsilon.$
	\end{enumerate}
\end{proposition}

\noindent
The proof of Theorem C will take place in two steps:

\begin{itemize}

\item[Step 1:] Using Proposition \ref{PROPNINE} and properties of the $\mathcal{F}^c$-holonomy to find Lebesgue full-measure subsets
of the diagonal on which we have good control of the Lyapunov exponents:

\begin{proposition}\label{SET}
         For any $\epsilon>0$, there is a Lebesgue full measure set $S^{+}_t(\epsilon)\subset S_t$, such that if $\phi \in S^{+}_t(\epsilon)$, then $\phi \in G_{\phi}(\epsilon)$.
\end{proposition}

\item[Step 2:] Use of Proposition \ref{SET} to prove Theorem C using a similar technique to the proof of the Ledrappier-Young formula
and then taking a countable intersection of Lebesgue full-measure sets.

\end{itemize}

\begin{proof}[Step 1: Proof of Proposition \ref{SET}:]

\begin{lemma}\label{THREEEPSILON}
	There exists $\delta>0$ such that if $|\phi_1-\phi_0|<\delta$ and $\theta_0 \in G_{\phi_0}(\epsilon/2)$, then $\theta_1:=g_{\phi_0,\phi_1}(\theta_0) \in G_{\phi_1} (\epsilon)$, where $g_{\phi_0,\phi_1}:\mathbb{T}_{\phi_0}\rightarrow \mathbb{T}_{\phi_1}$ is the holonomy transformation.
\end{lemma}

\begin{proof}
	We compute
	\begin{align*}
	\left|\frac{1}{n}\log (B_{\phi_1,t}^n)'(\theta_1)-\chi_{\phi_1,t}\right|
	\leq \left|\log \left(\frac{(B_{\phi_1,t}^n)'(\theta_1)}{ (B_{\phi_0,t}^n)'(\theta_0)}\right)^\frac{1}{n}\right|+\left|\frac{1}{n}\log (B_{\phi_0,t}^n)'(\theta_0)-\chi_{\phi_0,t}\right|+\left|\chi_{\phi_1,t}-\chi_{\phi_0,t}\right|.
	\end{align*}
	By Proposition \ref{EXPLICITLYAP}, $\chi_{\phi,t}$ is continuous in $\phi$. Meanwhile, since $(\phi_0,\theta_0)$ and $(\phi_1,\theta_1)$ are in the same leaf of $\mathcal{F}^c$, as in the proof of Proposition \ref{PROP:HOLDERCONJ}, the first term above can be made less than $\epsilon/4$ whenever $|\phi_0-\phi_1|$ is sufficiently small. Therefore, there exists $\delta >0$ so that 
	$$\left|\frac{1}{n}\log (B_{\phi_1,t}^n)'(\theta_1)-\chi_{\phi_1,t}\right| \leq 
	\left|\frac{1}{n}\log (B_{\phi_0,t}^n)'(\theta_0)-\chi_{\phi_0,t}\right| +\frac{\epsilon}{2}.$$
	Take the limit supremum on both sides, we get
	$$\limsup_{n\rightarrow \infty} \left|\frac{1}{n}\log (B_{\phi_1,t}^n)'(\theta_1)-\chi_{\phi_1,t}\right|<\epsilon.$$
\end{proof}

\begin{lemma}\label{HAUSDORFF2}
	Let $\phi_0 \in {\rm interior}(S_t)$. Then for all $\epsilon>0$, there exists $0<H_1<1$, $\delta>0$ such that if $|\phi-\phi_0|<\delta$, then $\rm{HD}(\mathbb{T}_\phi\setminus G_{\phi}(\epsilon))<H_1<1$.
\end{lemma}

\begin{proof}
	Denote $E_{\phi}(\epsilon):=\mathbb{T}_\phi \setminus G_\phi (\epsilon)$. By Proposition \ref{PROPNINE}, there exists $0<H_0<1$ satisfying \begin{equation*}\textrm{HD}(E_{\phi_0}(\epsilon/2))<H_0<1,\end{equation*}
	so that we can choose $H_0<H_1<1$ with $\gamma:=H_0/H_1<1$. By Proposition \ref{PROP:HOLDERCONJ}, there exists $\delta >0$ such that $|\phi-\phi_0|<\delta$ implies $g_{\phi_0,\phi}$ is H\"{o}lder with exponent $\gamma$, then by Lemma \ref{LEM:HOLDERHAUSDORFF}, we obtain an upper bound on the Hausdorff dimension of the set $g_{\phi_0,\phi} (E_{\phi_0}(\epsilon/2))$:
	
	\begin{equation*}\textrm{HD}(g_{\phi_0,\phi} (E_{\phi_0}(\epsilon/2)))<H_0/\gamma =H_1<1.\end{equation*}
	Meanwhile, by reducing $\delta$ if necessary, the previous lemma gives that
	$$g_{\phi_0,\phi}(G_{\phi_0}(\epsilon/2))\subset G_\phi (\epsilon),$$ hence $E_\phi (\epsilon) \subseteq g_{\phi_0,\phi} (E_{\phi_0}(\epsilon/2))$.
\end{proof}

We now use Lemmas \ref{THREEEPSILON} and \ref{HAUSDORFF2} to prove Proposition \ref{SET}.

	Let $X_t \Subset {\rm interior}(S_t)$ be a compact interval. By Lemma \ref{HAUSDORFF2}, there is a constant $0< H_2 < 1$ such that for all $\phi \in X_t$, we have $\textrm{HD}(E_\phi (\epsilon/2))<H_2$. Choose $0<  \gamma < 1$ so that $H_2/\gamma <1$, then by Lemma \ref{LEM:DIAGONALHOLDER}, there exists $\delta > 0$ (which implicitly depends on $\phi$) such that 
	$g_{I_{\phi, \delta},\Delta_{\phi, \delta}}$ is H\"{o}lder continuous with exponent $\gamma$, where $I_{\phi, \delta},\Delta_{\phi, \delta}$ are defined in Lemma \ref{LEM:DIAGONALHOLDER}.
	
	Next let us fix $\phi_0 \in X_t$. For any $\theta \in I_{\phi_0,\delta} \cap G_{\phi_0}(\epsilon/2)$, Lemma \ref{THREEEPSILON} immediately gives that  
	$g_{I_{\phi_0, \delta},\Delta_{\phi_0, \delta}}(\theta) \subset G_{\phi_1} (\epsilon),$ where $\phi_1 := g_{I_{\phi_0, \delta}, \Delta_{\phi_0, \delta}} (\theta)$. 
	
	Meanwhile, since $\textrm{HD}(E_{\phi_0}(\epsilon/2)) <1$ (Proposition \ref{PROPNINE}), by the discussion in the first part of the proof and Lemma \ref{LEM:HOLDERHAUSDORFF}, the set $g_{I_{\phi_0, \delta},\Delta_{\phi_0, \delta}}(I_{\phi_0, \delta} \cap E_{\phi_0}(\epsilon/2))$ also has Hausdorff dimension less than one, hence it has Lebesgue measure $0$. Therefore, its complement, $g_{I_{\phi_0, \delta},\Delta_{\phi_0, \delta}}(I_{\phi_0, \delta} \cap G_{\phi_0}(\epsilon/2))$, must have full Lebesgue measure in $\Delta_{\phi_0, \delta}$. Since $\phi_0$ and  the compact interval $X_t$ are arbitrary, we are done.

\end{proof}

\begin{proof}[Step 2: Proof of Theorem C:]
	Let $\phi \in S^{+}_t (\epsilon_0)$, where $S^{+}_t (\epsilon_0)$ are in Proposition \ref{SET}. Consider $\widehat{\Delta}_\delta:=(\phi-\delta,\phi+\delta) \subset \Delta_t$ and denote $\widehat{I}_\delta:=g_{\Delta,\phi}(\widehat{\Delta}_\delta)$ the holonomy image of $\widehat{\Delta}_\delta$ in $\mathbb{T}_\phi$.\par
	
	Let $B_{\phi,t}^n$ be the first iterate so that $2\pi \leq |\widetilde{B}_{\phi,t}^n(\widetilde{I}_\delta)|= 2\pi k$, where $\widetilde{B}_{\phi,t}:\mathbb{R}\rightarrow \mathbb{R}$ is a lift of $B_{\phi,t}$ and $\widetilde{I}_\delta$ is the lift of $\widehat{I}_\delta$. By the Intermediate Value Theorem, there exists $\xi_\delta \in \widehat{I}_\delta$ with 
	
	\begin{equation}\label{EQ2}
	2\pi\leq |(B^n_{\phi,t})'(\xi_\delta)|\cdot|\widehat{I}_\delta|\leq 2\pi k.\end{equation}
	
	On the other hand, by Distortion Control (Lemma \ref{DENJOY}), there is a constant $M$ so that for all $I_\delta$, 
	
	$$\frac{1}{M}\cdot |(B_{\phi,t}^n)'(\phi)|\leq |(B_{\phi,t}^n)'(\xi_\delta)|\leq M\cdot |(B_{\phi,t}^n)'(\phi)|.$$

	Taking logarithm in the above inequalities and combine with (\ref{EQ2}):
	\begin{equation}\label{EQ3}
	\log 2\pi - \log M - \log|(B_{\phi,t}^n)'(\phi)|\leq \log |\widehat{I}_\delta| 
	\leq \log 2\pi k +\log M - \log|(B_{\phi,t}^n)'(\phi)|.\end{equation}
	
	Next, since the MME satisfies $(B_{\phi,t})^* \eta_{\phi,t} = k\cdot \eta_{\phi,t},$ for all $\widehat{I}_\delta$, we have $\log \eta_{\phi,t}(\widehat{I}_\delta)=\log C_\delta -n\log k,$ where $C_\delta= |\widetilde{B}_{\phi,t}^n (\widetilde{I}_\delta)|$. Under the transverse invariant measure $\eta$ (as in Proposition \ref{TRANSINVARMAINPROP}), we have 
	$$\log \mu_t (\widehat{\Delta}_\delta)=\log \eta_{\phi,t}(\widehat{I}_\delta)=\log C_\delta -n\log k.$$
	
	Hence, combine with (\ref{EQ3}):
	
	\begin{equation}\label{EQUATION5}
	\frac{\log C_\delta-n\log k}{\log 2\pi k+\log M-\log|(B_{\phi,t}^n)'(\phi)|}\leq\frac{\log \mu_t (\widehat{\Delta}_\delta)}{\log |\widehat{I}_\delta|}\leq\frac{\log C_\delta-n\log k}{\log 2\pi-\log M -\log|(B_{\phi,t}^n)'(\phi)|}.\end{equation}
	
	Divide both numerator and denominator by $-n$, and notice that $n \rightarrow \infty$ as $\delta \rightarrow 0$, so by taking limit infimum and limit supremum respectively in the first and second inequalities above, we obtain
	
	\begin{equation}\label{EQ4}\frac{\log k}{\chi_{\nu_{\phi,t}}+\epsilon_0}
	\leq \liminf_{\delta \rightarrow 0} \frac{\log \mu_t (\widehat{\Delta}_\delta)}{\log |\widehat{I}_\delta|}
	\leq \limsup_{\delta \rightarrow 0} \frac{\log \mu_t (\widehat{\Delta}_\delta)}{\log |\widehat{I}_\delta|}\leq \frac{\log k}{\chi_{\nu_{\phi,t}}-\epsilon_0}.\end{equation}
	
	Since the H\"{o}lder exponent of $g_{\widehat{I}_\delta,\widehat{\Delta}_\delta}$ converges to $1$ when $\delta \rightarrow 0$, we have $\lim_{\delta\rightarrow 0} \frac{\log |\widehat{I}_\delta|}{\log | \widehat{\Delta}_\delta|}=1$, so in (\ref{EQ4}) we can replace $\log |\widehat{I}_\delta|$ by $\log |\widehat{\Delta}_\delta|$.
	Meanwhile, by choosing $\epsilon_0$ sufficiently small, we have
	
	\begin{equation}\label{EQUATION6}\frac{\log k}{\chi_{\nu_{\phi,t}}}-\epsilon
\leq \liminf_{\delta \rightarrow 0} \frac{\log \mu_t (\widehat{\Delta}_\delta)}{\log |\widehat{I}_\delta|}
\leq \limsup_{\delta \rightarrow 0} \frac{\log \mu_t (\widehat{\Delta}_\delta)}{\log |\widehat{I}_\delta|}\leq \frac{\log k}{\chi_{\nu_{\phi,t}}}+\epsilon.\end{equation}
Let
$$S^{+}_t:=\bigcap_{n=1}^\infty S^{+}_t(1/n).$$
The set $S^{+}_t$ satisfies the following:

\begin{enumerate}
    \item[{\rm (i)}] It has full Lebesgue measure, since each $S^{+}_t(1/n)$ has full Lebesgue measure (Proposition \ref{SET});
    \item[{\rm (ii)}] By (\ref{EQUATION6}), for each $\phi \in S^{+}_t$ the pointwise dimension of the Lee-Yang measure $\mu_t$ satisfies $d_{\mu_t}(\phi)=\log k/\chi_{\nu_{\phi,t}}$.
    \item[{\rm (iii)}] By the Ledrappier--Young Formula, $\chi_{\phi,t}=h_{\nu_{\phi,t}}(B_{\phi,t})<\log k$, so for $\phi \in S^{+}_t$ we have 
$d_{\mu_t}(\phi)>1$.
\end{enumerate}
\par

\vspace{0.1in}
Now, the construction of the dense set $S^{-}_t \subset S_t$ such that $\kappa_t(\phi) < 1$ for any $\phi \in S^{-}_t$ will be an easy adaptation of the previous discussion.
Consider any  $\phi \in S_t$. For $\eta_{\phi,t}$-almost every  $\theta \in \mathbb{T}_\phi$ (recall $\eta_{\phi,t}$ is the MME for the map $B_{\phi,t}:\mathbb{T}_\phi\rightarrow\mathbb{T}_\phi$), its pointwise Lyapunov exponent satisfies
\begin{equation}\label{SUPPORTMU}
\lim_{n\rightarrow \infty} \frac{1}{n} \log |(B_{\phi,t}^n)'(\theta)| = \chi_{\eta_{\phi,t}}(B_{\phi,t}) > \log k,\end{equation}
by Proposition \ref{PROP:HD} and the Ledrappier--Young Formula. Since $\eta_{\phi,t}$ is supported on the whole circle, the set of points $\theta \in \mathbb{T}_\phi$ satisfying Inequality (\ref{SUPPORTMU}) is dense. Hence, using that the Lyapunov exponent (with respect to the MME), $\chi_{\eta_{\phi,t}}(B_{\phi,t})$, is continuous in $\phi$ \cite{M}, and that the holonomy transformations are homeomorphisms onto their images, we conclude that there is a dense subset $S^-_t$ of $S_t$ such that each $\phi \in S^-_t$ satisfies
\begin{equation*}
\lim_{n\rightarrow \infty} \frac{1}{n} \log |(B_{\phi,t}^n)'(\phi)| > \log k.\end{equation*}
The result then follows using Inequalities (\ref{EQUATION5}) and taking $\delta \rightarrow 0$.
\end{proof}


\section{Proof of Theorem D}
\label{SEC:THMD}

This section is an adaptation of the clever techniques in
M\"uller-Hartmann's paper \cite{MULLERHARTMAN} to prove Theorem D.  Because of
the ``electrostatic representation'' (\ref{electrostat rep}) for the free energy,
we will present the result as a basic fact about the logarithmic potential of a
measure.  While the Lee--Yang measure $\mu_t$ is supported on the unit circle,
we find it convenient to perform a suitable M\"obius transformation to move the
support of the measure from $\mathbb{T}$ to $\mathbb{R}$, studying
\begin{align*}
f_\mu(z) := \int_{\mathbb{R}} \log|\zeta - z| d\mu(\zeta),
\end{align*}
where $\mu$ is a probability measure on $\mathbb{R}$.

\begin{proposition}[{\bf M\"uller-Hartmann \cite{MULLERHARTMAN}}]\label{DIFFERENTIABILITY}
Let $\mu$ be a probability measure on $\mathbb{R}$ which satisfies
$$\kappa \equiv d_\mu(0) = \lim_{\delta\rightarrow 0} \frac{\log \mu ([-\delta,+\delta])}{\log 2\delta} > 0.$$ 
Then, there exists a real-analytic function $f_{{\rm reg}}(y)$ such that 
\begin{align}\label{EQN:CRITICAL_EXPONENT_LIMIT}
\lim_{y \rightarrow 0} \frac{\log \left|f_\mu(iy) - f_{{\rm reg}}(y)\right|}{\log |y|} = \kappa.
\end{align}
In other words, $f_\mu$ has critical exponent $\kappa$ when crossing $\mathbb{R}$ vertically.
\end{proposition}

\noindent
We will need the following integration by parts formula and include a proof
here, for lack of a convenient reference.  

\begin{lemma}[{\bf Integration by Parts}]\label{byparts}
	Let $\Phi_\mu(t)=\mu ([a,t])$ and $f:[a,b]\rightarrow \mathbb{R}$ be a differentiable function so that $f'(t) \in L^1({\rm Leb}([a,b]))$. Then we have:
	$$\int_a^b f(t)d\mu(t)=f(b)\Phi_\mu(b)-\int_a^b f'(t) \Phi_\mu (t)dt.$$
	
\end{lemma}

\begin{proof}
	We have
	$$\Phi_\mu (t)=\int_a^t d\mu = \int_a^b 1_{[a,t]}(s)d\mu(s),$$
	Hence,
	
	\begin{align*} \int_a^b f'(t)\Phi_\mu(t) dt &=\int_a^b \int_a^b f'(t)1_{[a,t]}(s)d\mu(s) dt 
		=\int_a^b \int_a^b f'(t)1_{[a,t]}(s)dt d\mu(s) \\
		&=\int_a^b \int_a^b f'(t)1_{[s,b]}(t)dt d\mu(s) = \int_a^b\int_s^b f'(t) dt d\mu(s)\\
&= \int_a^b f(b)-f(s) d\mu(s) = f(b)\Phi_\mu (b)-\int_a^b f(s) d\mu(s).
	\end{align*}

\noindent
The second equality uses Fubini's Theorem, which is allowed since since $f' \in L^1({\rm Leb}([a,b]))$ implies
that $f'(t)1_{[a,t]}(s) \in L^1({\rm Leb}([a,b]) \times \mu)$, and
the third equality uses that $1_{[a,t]}(s)=1$ iff $a\leq s \leq t \leq b$
iff $1_{[s,b]}(t)=1$. 

\end{proof}

\begin{proof}[Proof of Proposition \ref{DIFFERENTIABILITY}:]
It suffices to show
that there exists real-analytic $f_{\rm reg}(y)$ such that for any $\epsilon >
0$ there are constants $C_1, C_2 >0$ so that for all $y \neq 0$ sufficiently close to
$0$ the logarithmic potential $f_\mu$ satisfies
\begin{align*}
C_1 |y|^{\kappa+\epsilon} \leq |f_\mu(iy)-f_{{\rm reg}}(y)| \leq C_2 |y|^{\kappa-\epsilon}.
\end{align*}
By assumption on $\mu$, given $\epsilon >0$, there exists $\delta_0, C>0$ such that whenever $0<\delta<\delta_0$, 
\begin{equation}\label{smallmeasure} C\delta^{\epsilon+\kappa}<\mu([-\delta,+\delta])<C\delta^{-\epsilon+\kappa}.\end{equation}
Since $$\int_{\mathbb{R} \setminus [-\delta_0,\delta_0]} \log|\zeta - iy| d\mu(\zeta)$$ is real analytic in $y$,
it is enough to consider $f_\mu$ as an integral on $[-\delta_0,\delta_0]$. 

Separate $f_\mu$ as a sum of two integrals:
\begin{align*}
f_\mu(iy)&=\int_0^{\delta_0} \log |\zeta-iy|d\mu(\zeta)+\int_{-\delta_0}^0 \log|\zeta-iy|d\mu(\zeta)\\
&=\frac{1}{2} \left(\int_0^{\delta_0} \log (\zeta^2+y^2)d\mu(\zeta)+\int_{-\delta_0}^0 \log (\zeta^2+y^2)d\mu(\zeta)\right).
\end{align*}
Applying Integration by Parts (Lemma \ref{byparts}) to the first integral gives
\begin{align*}
\int_0^{\delta_0}\log(\zeta^2+y^2)d\mu(\zeta)=\log(y^2+\delta_0^2)\Phi_\mu(\delta_0)-2\int_0^{\delta_0}\frac{\zeta \Phi_\mu(\zeta)}{\zeta^2+y^2}  d\zeta,
\end{align*}
where $\Phi_\mu(\zeta) := \mu([0,\zeta])$.
By a change of variable $h(\psi)=-\psi\equiv \zeta$, the second integral becomes
$$\int_{-\delta_0}^0 \log (\zeta^2+y^2)d\mu(\zeta)=\int_{-\delta_0}^0 \log ((-\psi)^2+y^2) d\mu(h(\psi))=\int_0^{\delta_0} \log (\psi^2+y^2)d(h_*\mu)(\psi).$$
Again, Integration by Parts gives
\begin{align*}
\int_{0}^{\delta_0} \log (\psi^2+y^2)d\mu(\psi)=\log(y^2+\delta_0^2)\Phi_\nu(\delta_0)-2\int_0^{\delta_0} \frac{\psi \Phi_\nu(\psi)}{\psi^2+y^2}d\psi,
\end{align*}
where $\Phi_\nu(\zeta) = \nu([0,\zeta]) = \mu([-\zeta,0])$.
Therefore, we find
\begin{align*}
f_\mu(iy) = \int_{-\delta_0}^{\delta_0} \log|\zeta -iy| d\mu(\zeta) = \frac{1}{2} \int_{-\delta_0}^{\delta_0} \log(\zeta^2+y^2) d\mu(\zeta) = \log(y^2+\delta_0^2)\Phi(\delta_0) - \int_{0}^{\delta_0} \frac{\zeta \Phi(\zeta)}{\zeta^2+y^2}d\zeta,
\end{align*}
where $\Phi(\zeta) := \Phi_\mu(\zeta) + \Phi_{\nu}(\zeta) = \mu([-\zeta,\zeta])$.

By assumption $\delta_0 \neq 0$, so the term $\log(y^2+\delta_0^2)\Phi(\delta_0)$ is analytic for all $y \in \mathbb{R}$, therefore it remains to consider the term
\begin{equation*}h(y):=\int_{0}^{\delta_0} \frac{\zeta \Phi(\zeta)}{\zeta^2+y^2}d\zeta.\end{equation*}
There is a non-negative integer $m$ satisfying $2m < \kappa \leq 2m+2$. We claim that the singular part $h_{{\rm sing}}:=h-h_{{\rm reg}}$ can be expressed as
$$h_{{\rm sing}}(y)=\int_0^{\delta_0} \frac{\zeta \Phi(\zeta)}{\zeta^2+y^2}\left(\frac{iy}{\zeta}\right)^{2m+2} d\zeta.$$
To see this, consider the difference:
$$h_{{\rm reg}}(y):=\int_0^{\delta_0}\frac{\zeta \Phi(\zeta)}{\zeta^2+y^2}-\frac{\zeta \Phi(\zeta)}{\zeta^2+y^2}\left(\frac{iy}{\zeta}\right)^{2m+2}d\zeta= \sum_{j=0}^m -y^{2j}\int_0^{\delta_0}\frac{\Phi(\zeta)}{\zeta^{2j+1}}d\zeta. $$
By the choice of $m$, we can choose $\epsilon > 0$ sufficiently small that $2m+1-\kappa+\epsilon < 1$.  Therefore, each integral in the sum satisfies
$$\int_0^{\delta_0}\frac{\Phi(\zeta)}{\zeta^{2j+1}}d\zeta<\int_0^{\delta_0}\frac{C\zeta^{-\epsilon+\kappa}}{\zeta^{2j+1}}d\zeta=\int_0^{\delta_0}\frac{C}{\zeta^{2j+1-\kappa+\epsilon}}d\zeta<\infty,$$
where the first inequality uses (\ref{smallmeasure}). We conclude that $h_{{\rm reg}}$ is a polynomial in $y$. Hence the proof reduces to studying $h_{{\rm sing}}$. The following estimate is due to Inequalities (\ref{smallmeasure}).
\begin{equation}\label{eqa1}C \int_0^{\delta_0} \frac{\zeta^{\epsilon+\kappa+1}}{\zeta^2+y^2}\left(\frac{y}{\zeta}\right)^{2m+2}d\zeta<|h_{{\rm sing}}(y)|
<C\int_0^{\delta_0} \frac{\zeta^{\kappa-\epsilon+1}}{\zeta^2+y^2}\left(\frac{y}{\zeta}\right)^{2m+2} d\zeta.\end{equation}
By the change of variable $\zeta=|y|\exp(\tau)$, Inequalities (\ref{eqa1}) becomes
\begin{equation}\label{impliesdiff}
	C_1\cdot |y|^{\kappa+\epsilon}<|h_{{\rm sing}}(y)|<C_2 \cdot |y|^{\kappa-\epsilon},\end{equation}
where $$C_1:=C\int_{-\infty}^{0} \frac{\exp(\tau(\kappa+\epsilon-2m))}{\exp(2\tau)+1}d\tau \ \ \ \ \mbox{and} \ \ \ \  C_2:=C\int_{-\infty}^{\infty} \frac{\exp(\tau(\kappa-\epsilon-2m))}{\exp(2\tau)+1}d\tau$$ are finite constants independent of all $y$ with $|y|<\delta_0$. This proves the assertion.

\end{proof}


\appendix

\section{Renormalization for Lee--Yang zeros on the Cayley Tree}
\label{SEC:APPENDIX_DERIVATION}

Even though Proposition \ref{PROP:RENORM} is proved in the papers by
M\"uller-Hartmann \cite{MULLERHARTMAN} and Barata--Marchetti
\cite{BARATAMARCHETTI}, we include a proof here for completeness.  We focus on
branching number $k=2$, leaving the straightforward generalization for
arbitrary $k$ to the reader.

Let us start with the rooted tree $\Gamma_n$.
For each $n$, let $r$ denote the root vertex of $\Gamma_n$ and consider the conditional partition functions
\begin{align}\label{EQN:CONDITIONAL_PF}
Z_{n}^+ \equiv Z_{n}^+(z,t) := \sum_{\substack{\sigma \, \, \mbox{\small such that} \\ \sigma(r) = +1}} W_n(\sigma) \qquad \mbox{and} \qquad  Z_{n}^- \equiv Z_{n}^-(z,t) := \sum_{\substack{\sigma \, \, \mbox{\small such that}  \\ \sigma(r) = -1}} W_n(\sigma),
\end{align}
where $W_n(\sigma) = e^{-\frac{H_n(\sigma)}{T}}$ and $H_n(\sigma)$ is the Hamiltonian given in Equation (\ref{EQN:HAMILTONIAN}).
By definition, the full partition function is $Z_n = Z_{n}^+ + Z_{n}^-$. 

We will first produce a recursion on $Z_{n}^+$ and $Z_{n}^+$, from which the statement of Proposition \ref{PROP:RENORM} will quickly follow.
As initial condition of the recursion, note that $\Gamma_0$ consists of just the root vertex~$r$, so that $Z_0^+ = z^{-\frac{1}{2}}$ and $Z_0^- = z^{\frac{1}{2}}$.

Now consider arbitrary $n \geq 0$.   The tree $\Gamma_{n+1}$  is formed by
taking two copies of $\Gamma_n$  and attaching each of their root vertices by an
edge to the root vertex of $\Gamma_{n+1}$.  Let us call the two sub-trees the
``left tree'' $\Gamma_n^L$ and the ``right tree'' $\Gamma_n^R$, and let us
denote their root vertices by $r^L$ and $r^R$, respectively.

To express $Z_{n+1}^+$ in terms of $Z_{n}^+$ and $Z_{n}^+$, we consider all spin configurations $\sigma: V_{n+1}\rightarrow \{\pm 1\}$ with $\sigma(v) = +1$.  
Given such a spin configuration, let us denote the restriction of $\sigma$ to the vertex set of $\Gamma_n^L$ by $\sigma^L$ and the restriction
of $\sigma$ to the vertex set of $\Gamma_n^R$ by $\sigma^R$.   We then have
\begin{align*}
H_{n+1}(\sigma) = -h -J(\sigma(r^L) + \sigma(r^R)) + H_n(\sigma^L) + H_n(\sigma^R),
\end{align*}
whose Gibbs Weight is
\begin{align*}
W_{n+1}(\sigma) = z^{-\frac{1}{2}} t^{-\frac{1}{2}(\sigma(r^L) + \sigma(r^R))} \, W_n(\sigma^L) \,  W_n(\sigma^R).
\end{align*}
Summing over all configurations $\sigma: V_{n+1} \rightarrow \{\pm 1\}$ with $\sigma(r) = +1$ we find
\begin{align*}
Z_{n+1}^+ = z^{-\frac{1}{2}} \Big(t^{-1} (Z_{n}^+)^2 + 2 \, Z_{n}^+ Z_{n}^- + t \, (Z_{n}^-)^2\Big) = z^{-\frac{1}{2}}\Big(t^{-\frac{1}{2}} Z_{n}^+ + t^{\frac{1}{2}} Z_{n}^-\Big)^2,
\end{align*}
where the first summand (of the second expression) corresponds to those $\sigma$ with $\sigma(r^L) = \sigma(r^R) = +1$, the second to those with $\sigma(r^L) = -\sigma(r^R)$, and
the third to those with $\sigma(r^L) = \sigma(r^R) = -1$.  Similarly,
\begin{align*}
Z_{n+1}^- = z^{\frac{1}{2}}\Big(t^{\frac{1}{2}} Z_{n}^+ + t^{-\frac{1}{2}} Z_{n}^-\Big)^2.
\end{align*}
If we let $w_n = Z_{n}^- / Z_{n}^+$ then we have the recursion
\begin{align*}
w_{n+1} = \frac{z^{\frac{1}{2}}\Big(t^{\frac{1}{2}} Z_{n}^+ + t^{-\frac{1}{2}} Z_{n}^-\Big)^2}{z^{-\frac{1}{2}}\Big(t^{-\frac{1}{2}} Z_{n}^+ + t^{\frac{1}{2}} Z_{n}^-\Big)^2} = z \left(\frac{t+w_n}{1+tw_n}\right)^2 = B_{z,t}(w_n).
\end{align*}
Using that $Z_0^+ = z^{-\frac{1}{2}}$ and $Z_0^- = z^{\frac{1}{2}}$, we find $w_0 = z$.  Meanwhile, $Z_n(z,t) = Z_{n}^+ + Z_{n}^- = 0$ if and only if $w_n = -1$, so the Lee--Yang zeros
for the rooted tree $\Gamma_n$ are solutions to $w_n = B_{z,t}^n(z) = -1.$

\vspace{0.1in}

We now discuss how to adapt the formula for the full (unrooted) Cayley Tree
$\widehat{\Gamma}_n$.  Let $c$ denote the ``central vertex'' of
$\widehat{\Gamma_n}$, i.e.\ the vertex at distance $n$ from the leaves of
$\widehat{\Gamma}_n$ and let $\widehat Z_{n}^+$ and $\widehat Z_{n}^-$ denote
the conditional partition functions, conditioned on $\sigma(c) = +1$ and
$\sigma(c) = -1$, respectively (analogous to Equation
(\ref{EQN:CONDITIONAL_PF})).  Since we consider branching number $k=2$,
$\widehat{\Gamma}_n$ is obtained by taking three copies of the rooted tree
$\Gamma_{n-1}$ and attaching each of their root vertices by an edge to the
central vertex $c$.  In essentially the same way as for the rooted tree above,
one finds that
\begin{align*}
\widehat Z_{n}^+ = z^{-\frac{1}{2}}\Big(t^{-\frac{1}{2}} Z_{n-1}^+ + t^{\frac{1}{2}} Z_{n-1}^-\Big)^3 \qquad \mbox{and} \qquad \widehat Z_{n}^- = z^{\frac{1}{2}}\Big(t^{\frac{1}{2}} Z_{n-1}^+ + t^{-\frac{1}{2}} Z_{n-1}^-\Big)^3.
\end{align*}
Therefore, $\widehat Z_n = \widehat Z_{n}^+ + \widehat Z_{n}^- = 0$ if and only if
\begin{align*}
\frac{\widehat Z_{n}^+}{\widehat Z_{n}^-} = B_{z,t,3}(w_{n-1}) = B_{z,t,3} \circ B_{z,t,2}^{n-1}(z) = -1.
\end{align*}

\qed (Proposition \ref{PROP:RENORM}.)


\vspace{0.2in}


\begin{thebibliography}{10}

\bibitem{BARATAGOLDBAUM}
J.~Barata and P.~Goldbaum.
\newblock On the {D}istribution and {G}ap {S}tructure of {L}ee-{Y}ang {Z}ero
  {I}sing {M}odel: {P}eriodic and {A}periodic {C}ouplings.
\newblock {\em Jour. Stat. Phys}, 103:857--891, 2001.

\bibitem{BARATAMARCHETTI}
J.~Barata and D.~Marchetti.
\newblock Griffiths' singularities in diluted {I}sing {M}odels on the {C}ayley
  tree.
\newblock {\em Jour. Stat. Phys}, 88:231--268, 1997.

\bibitem{BAXTER}
Rodney~J. Baxter.
\newblock {\em Exactly solved models in statistical mechanics}.
\newblock Academic Press, Inc. [Harcourt Brace Jovanovich, Publishers], London,
  1989.
\newblock Reprint of the 1982 original.

\bibitem{BR}
Lipman Bers and H.~L. Royden.
\newblock Holomorphic families of injections.
\newblock {\em Acta Math.}, 157(3-4):259--286, 1986.

\bibitem{BBCKK}
M.~Biskup, C.~Borgs, J.~T. Chayes, L.~J. Kleinwaks, and R.~Koteck\'y.
\newblock Partition function zeros at first-order phase transitions: a general
  analysis.
\newblock {\em Comm. Math. Phys.}, 251(1):79--131, 2004.

\bibitem{BLR1}
P.~Bleher, M.~Lyubich. R.~Roeder
\newblock Lee--{Y}ang zeros for the {DHL} and 2{D} rational dynamics, {I}.
  {F}oliation of the physical cylinder.
\newblock {\em J. Math. Pures Appl. (9)}, 107(5):491--590, 2017.

\bibitem{Car}
John~L. Cardy.
\newblock Conformal invariance and the {Y}ang-{L}ee edge singularity in two
  dimensions.
\newblock {\em Phys. Rev. Lett.}, 54:1354--1356, Apr 1985.

\bibitem{Erchenko}
Alena Erchenko.
\newblock Flexibility of exponents for expanding maps on a circle.
\newblock To appear in Discrete and Continuous Dynamical Systems. See also, arXiv preprint
\url{https://arxiv.org/pdf/1704.00832.pdf}.

\bibitem{Fis1}
Michael~E. Fisher.
\newblock Yang-{L}ee edge singularity and ${\ensuremath{\phi}}^{3}$ field
  theory.
\newblock {\em Phys. Rev. Lett.}, 40:1610--1613, Jun 1978.

\bibitem{GP}
Stefano Galatolo and Mark Pollicott.
\newblock Controlling the statistical properties of expanding maps.
\newblock {\em Nonlinearity}, 30(7):2737--2751, 2017.

\bibitem{IKS}
Yuli~S. Ilyashenko, Victor~A. Kleptsyn, and Petr Saltykov.
\newblock Openness of the set of boundary preserving maps of an annulus with
  intermingled attracting basins.
\newblock {\em J. Fixed Point Theory Appl.}, 3(2):449--463, 2008.

\bibitem{IVRII}
O.~Ivrii.
\newblock The geometry of the Weil-Petersson metric in complex dynamics., 2015.
\newblock To appear in Transactions of the AMS. See also, arXiv preprint
\url{https://arxiv.org/abs/1503.02590}.

\bibitem{Kifer}
Yuri Kifer.
\newblock Large deviations in dynamical systems and stochastic processes.
\newblock {\em Trans. Amer. Math. Soc.}, 321(2):505--524, 1990.

\bibitem{Kleptsyn}
Victor Kleptsyn, Dmitry Ryzhov, and Stanislav Minkov.
\newblock Special ergodic theorems and dynamical large deviations.
\newblock {\em Nonlinearity}, 25(11):3189--3196, 2012.

\bibitem{Krzy}
K.~Krzy\.zewski.
\newblock Some results on expanding mappings.
\newblock pages 205--218. Ast\'erisque, No. 50, 1977.

\bibitem{LEEYANG1}
T.~Lee and C.~Yang.
\newblock Statistical {T}heory of {E}quations of {S}tate and {P}hase
  {T}ransitions. {I}. {T}heory of {C}ondensation.
\newblock {\em Physical Review}, 87(3), 1952.

\bibitem{M}
Ricardo Ma\~n\'e.
\newblock The {H}ausdorff dimension of invariant probabilities of rational
  maps.
\newblock In {\em Dynamical systems, {V}alparaiso 1986}, volume 1331 of {\em
  Lecture Notes in Math.}, pages 86--117. Springer, Berlin, 1988.

\bibitem{AManning}
Anthony Manning.
\newblock A relation between {L}yapunov exponents, {H}ausdorff dimension and
  entropy.
\newblock {\em Ergodic Theory Dynamical Systems}, 1(4):451--459 (1982), 1981.

\bibitem{MCMULLEN}
C.~McMullen.
\newblock Dynamics on the {U}nit {D}isk: short geodesics and simple cycles.
\newblock {\em Commentarii Mathematici Helvetici}, 85(4):723--749, 2010.

\bibitem{MILNOR}
J.~Milnor.
\newblock Dynamics in {O}ne {C}omplex {V}ariable. ({A}{M}160:{A}{M}160.
\newblock {\em Princeton University Press}, 2011.

\bibitem{milnor1}
John Milnor.
\newblock Fubini foiled: {K}atok's paradoxical example in measure theory.
\newblock {\em Math. Intelligencer}, 19(2):30--32, 1997.

\bibitem{MULLERHARTMAN}
E~M\"uller-Hartmann.
\newblock Theory of the {I}sing model on a {C}ayley {T}ree.
\newblock {\em Z. Physik B}, 27:161--168, 1977.

\bibitem{MULLERZITTARZ}
E.~M\"uller-Hartmann and J.~Zittartz.
\newblock Phase {T}ransitions of {C}ontinuous {O}rder: {I}sing model on a
  {C}ayley tree.
\newblock {\em Z Physik B}, 22(59), 1975.

\bibitem{MBT}
G~Mussardo, R~Bonsignori, and A~Trombettoni.
\newblock Yang–lee zeros of the yang–lee model.
\newblock {\em Journal of Physics A: Mathematical and Theoretical},
  50(48):484003, 2017.


\bibitem{PR}
Han Peters and Guus Regts.
\newblock Location of zeros for the partition function of the Ising model on bounded degree graphs.
\newblock Preprint, see \url{https://arxiv.org/abs/1810.01699}.


\bibitem{PUJALS}
Enrique~R. Pujals, Leonel Robert, and Michael Shub.
\newblock Expanding maps of the circle rerevisited: positive {L}yapunov
  exponents in a rich family.
\newblock {\em Ergodic Theory Dynam. Systems}, 26(6):1931--1937, 2006.

\bibitem{RUELLE_PR}
David Ruelle.
\newblock Extension of the {L}ee-{Y}ang circle theorem.
\newblock {\em Phys. Rev. Lett.}, 26:303--304, 1971.

\bibitem{RI}
David Ruelle.
\newblock An inequality for the entropy of differentiable maps.
\newblock {\em Bol. Soc. Brasil. Mat.}, 9(1):83--87, 1978.

\bibitem{Sacksteder}
Richard Sacksteder.
\newblock The measures invariant under an expanding map.
\newblock pages 179--194. Lecture Notes in Math., Vol. 392, 1974.

\bibitem{S}
Michael Shub.
\newblock Endomorphisms of compact differentiable manifolds.
\newblock {\em Amer. J. Math.}, 91:175--199, 1969.

\bibitem{PS}
Michael Shub and Dennis Sullivan.
\newblock Expanding endomorphisms of the circle revisited.
\newblock {\em Ergodic Theory Dynam. Systems}, 5(2):285--289, 1985.

\bibitem{VANHOVE}
L.~Van-Hove.
\newblock Quelques prop\'et\'es g\'en\'erales de l'int\'egral de configuration
  d'un syst\'em de particles avec interaction.
\newblock {\em Physica}, 15:951--961, 1949.

\bibitem{RT}
Marcelo Viana and Krerley Oliveira.
\newblock {\em Foundations of ergodic theory}, volume 151 of {\em Cambridge
  Studies in Advanced Mathematics}.
\newblock Cambridge University Press, Cambridge, 2016.

\bibitem{Y1}
Lai-Sang Young.
\newblock Large deviations in dynamical systems.
\newblock {\em Trans. Amer. Math. Soc.}, 318(2):525--543, 1990.

\bibitem{Y}
Lai-Sang Young.
\newblock Ergodic theory of differentiable dynamical systems.
\newblock In {\em Real and complex dynamical systems ({H}iller\o d, 1993)},
  volume 464 of {\em NATO Adv. Sci. Inst. Ser. C Math. Phys. Sci.}, pages
  293--336. Kluwer Acad. Publ., Dordrecht, 1995.

\bibitem{Z}
Anna Zdunik.
\newblock Parabolic orbifolds and the dimension of the maximal measure for
  rational maps.
\newblock {\em Invent. Math.}, 99(3):627--649, 1990.

\end{thebibliography}
\end{document}